\documentclass[prd,reprint, nofootinbib,amsmath,amssymb, aps, floatfix]{revtex4-1}
\usepackage{dcolumn}% Align table columns on decimal point

\usepackage{bm}% bold math
\usepackage{amssymb,amsmath,amsthm,graphicx,ulem}
\usepackage[linktoc=page]{hyperref}

\usepackage{graphicx,subfigure}
\usepackage{epsfig}
\usepackage{amsmath}
\usepackage{amsfonts}
\usepackage{amssymb}
\usepackage[usenames]{color}
\usepackage{enumerate}

\usepackage[T1]{fontenc}

\usepackage{mathtools,xparse}

\DeclarePairedDelimiter{\norm}{\lVert}{\rVert}

\numberwithin{equation}{section}
\begin{document}
\normalem

\theoremstyle{definition}
\newtheorem{theorem}{Theorem}
\newtheorem{definition}{Definition}
\newtheorem{lemma}{Lemma}
\newtheorem{corollary}{Corollary}

\def\bra#1{\langle #1 |}
\def\ket#1{| #1 \rangle}
\def\inner#1#2{\langle #1 | #2 \rangle}

\def\ext{\textrm{ext} \: }
\def\rep{\textrm{rep} }
\def\Mm{\textrm{Mm}}

\title{A Holographic Entanglement Entropy Conjecture for General Spacetimes}
\author{Fabio Sanches}
\email{fabios@berkeley.edu} 
\author{Sean J. Weinberg}%
 \email{sjweinberg@berkeley.edu}
\affiliation{Berkeley Center for Theoretical Physics and Department of Physics\\
University of California, Berkeley, CA 94720, USA }%

\bibliographystyle{utcaps}

\begin{abstract}
We present a natural generalization of holographic entanglement entropy proposals beyond the scope of AdS/CFT by anchoring extremal surfaces to holographic screens.  Holographic screens are a natural extension of the AdS boundary to arbitrary spacetimes and are preferred codimension 1 surfaces from the viewpoint of the covariant entropy bound.  A broad class of screens have a unique preferred foliation into codimension 2 surfaces called leaves.  Our proposal is to find the areas of extremal surfaces achored to the boundaries of regions in leaves.  We show that the properties of holographic screens are sufficient to prove, under generic conditions, that extremal surfaces anchored in this way always lie within a causal region associated with a given leaf.  Within this causal region, a maximin construction similar to that of Wall proves that our proposed quantity satisfies standard properties of entanglement entropy like strong subadditivity.  We conjecture that our prescription computes entanglement entropies in quantum states that holographically define arbitrary spacetimes, including those in a cosmological setting with no obvious boundary on which to anchor extremal surfaces.
\end{abstract}

\maketitle

\tableofcontents

%%%%%%%%%%%%%%%%
%%%%%%%%%%%%%%%%

\section{Introduction}
\label{intro}

A theory of quantum gravity should not apply only to asymptotically locally anti-de Sitter (AlAdS) spacetimes.  For this reason, the AdS/CFT correspondence   \cite{Maldacena:1997re, Witten:1998qj}, although immensely successful, has fallen short of a description of the quantum mechanics of spacetime.  The AdS restriction is severe: Maldacena's conjecture does not apply in an obvious way to even the cosmological spacetime we find ourselves in.

If a quantum theory applies to general spacetimes, it is desirable that it reduces to AdS/CFT in the appropriate cases.  This suggests a strategy for guessing properties of a complete theory: consider specific aspects of AdS/CFT and devise generalizations that are applicable to other spacetimes.  If one knew only of Special Relativity, she could guess aspects of General Relativity by thinking to ``promote'' the flat metric to a dynamical one.  Similar statements can be made about the relation between many other pairs of theories.  But retrospective examples obscure the challenge: one cannot confidently know what to promote (and how to promote it) to enlarge the regime of validity of a given theory.

Holographic entanglement entropy, proposed by  Ryu and Takayanagi (RT) \cite{Ryu:2006bv}, proved by Lewkowycz and  Maldacena \cite{Lewkowycz:2013nqa}, and made covariant by Hubeny, Rangamani, and Takayanagi (HRT) \cite{Hubeny:2007xt}, is a beautiful property (or, in the covariant case, conjecture) of AdS/CFT.  Below we describe a promotion of holographic entanglement entropy beyond the scope of AdS/CFT that applies just as well to cosmological spacetimes as it does to asymptotically AdS spacetimes.  In the case of the latter, it reduces to the HRT proposal.  Moreover, the promoted holographic entanglement entropy satisfies, for nontrivial reasons, expected properties of entanglement entropy like strong subadditivity.

The HRT prescription provides a way to compute entanglement entropy of a spatial region $A$ in a quantum state dual to an AlAdS spacetime.  The procedure is to consider $\partial A$, the boundary of the spatial region, and to find the area of a codimension 2 extremal surface that is anchored to $\partial A$.  A na\"{\i}ve extension of this idea to general spacetimes would be to take $A$ to be a region in the conformal boundary of an arbitrary spacetime.  This approach fails: what is the boundary of a closed FRW universe with past and future singularities?

In our proposal, we anchor extremal surfaces to a \emph{holographic screen}.  Holographic screens are codimension 1 surfaces that appear to be the most natural place for quantum states dual to arbitrary geometries to live on.  In fact, they were proposed by Bousso \cite{Bousso:1999cb} in an attempt to find the analogue of the AdS boundary when extending holography to general spacetimes.  If one believes the covariant entropy bound \cite{Bousso:1999xy}, then there is essentially no other reasonable class of surfaces for this purpose.\\

\emph{Outline}.  In section \ref{section2} we first review the concept of holographic screens \cite{Bousso:1999cb} with an emphasis on the recent developments of Bousso and Engelhardt \cite{Bousso:2015mqa,Bousso:2015qqa} which identified a class of screens that satisfy an area monotonicity law. We then give the definition of \emph{holographic screen entanglement entropy} and list a number of its key properties.  We conclude the section by stating our \emph{screen entanglement conjecture}---a proposal that holographic screen entanglement entropy actually measures von Neumann entropy in a putative holographic description of general spacetimes.  Section \ref{sec_SSA} contains technical developments including proofs of the properties of screen entanglement entropy that are advertised in section \ref{section2}.  Section \ref{sec_cosmology} gives cosmological examples of holographic screens and their extremal surfaces.  We focus particularly on FRW universes that approach 
de Sitter space at late times.
Section \ref{sec_conclusion} concludes by reviewing the procedure for computing screen entanglement entropy and by suggesting extensions to our proposal
such as possible methods for computing subleading contributions to holographic screen entanglement entropy.

\section{Holographic Screen Entanglement Entropy}
\label{section2}

%%%%%%%%%%%%%%% MODIFIED APPENDIX
%%%%%%%%%%%%%%%

We open this section with a brief review of holographic screens, especially past and future holographic screens.  Readers that are already familiar with the content \cite{Bousso:1999xy,Bousso:1999cb,Bousso:2015mqa,Bousso:2015qqa} may still find it useful to read through these paragraphs to become familiar with our conventions and notation.  Throughout this paper we will work in a globally hyperbolic spacetime $M$ of dimension $d$ that satisfies the null energy condition.  We assume that the spacetime satisfies the genericity conditions laid out in \cite{Bousso:1999cb,Bousso:2015qqa}.

Suppose that $B$ is an orientable spacelike codimension $2$ submanifold of $M$.  It is possible to find an independent pair of future directed null vector fields on $B$ that are everywhere orthogonal to $B$.  If one of these vector fields has vanishing null expansion on $B$, we will say that $B$ is \emph{marginal}.  If one vector field has 
zero expansion on $B$ while the other has negative (positive) expansion on $B$, we say that $B$ is \emph{marginally trapped (marginally anti-trapped)}.

A \textit{past holographic screen}  is a codimension-1 submanifold $\mathcal{H}$ of the spacetime that is foliated by marginally anti-trapped  compact spacelike surfaces called \textit{leaves}. The foliation into leaves is unique: other splittings of $\mathcal{H}$  cannot satisfy the marginally anti-trapped condition.  A \textit{future holographic screen} is instead foliated by marginally trapped surfaces. In this paper, we will always assume that leaves have the topology of $S^{d-2}$.

Holographic screens are generated by null foliations:  if $\{N_r\}$ is a null foliation of a spacetime, it is possible to identify a family of leaves $\{\sigma(r)\}$ with $\sigma(r) \subset N_r$ by finding the codimension 2 surface of maximal area on each null surface.  In general, this will break the values of the parameter $r$ into open intervals, some of which correspond to past holographic screens and others corresponding to future screens.\footnote{It is also possible that for some values of $r$, $\sigma(r)$ does not have a definite sign for $\theta^l$.  We leave the investigation of this scenario to future work.}  Isolated values of $r$ that lie between past and future screens correspond to the case where $\sigma(r)$ is an isolated extremal sphere which can join a past and future screen.  Such a sphere will not be considered to lie on a past or future holographic screen by convention.  This occurs in the case of a closed universe with a big crunch: see figure \ref{fig_big_crunch}.

\begin{figure}
\centering
\includegraphics[width=8cm]{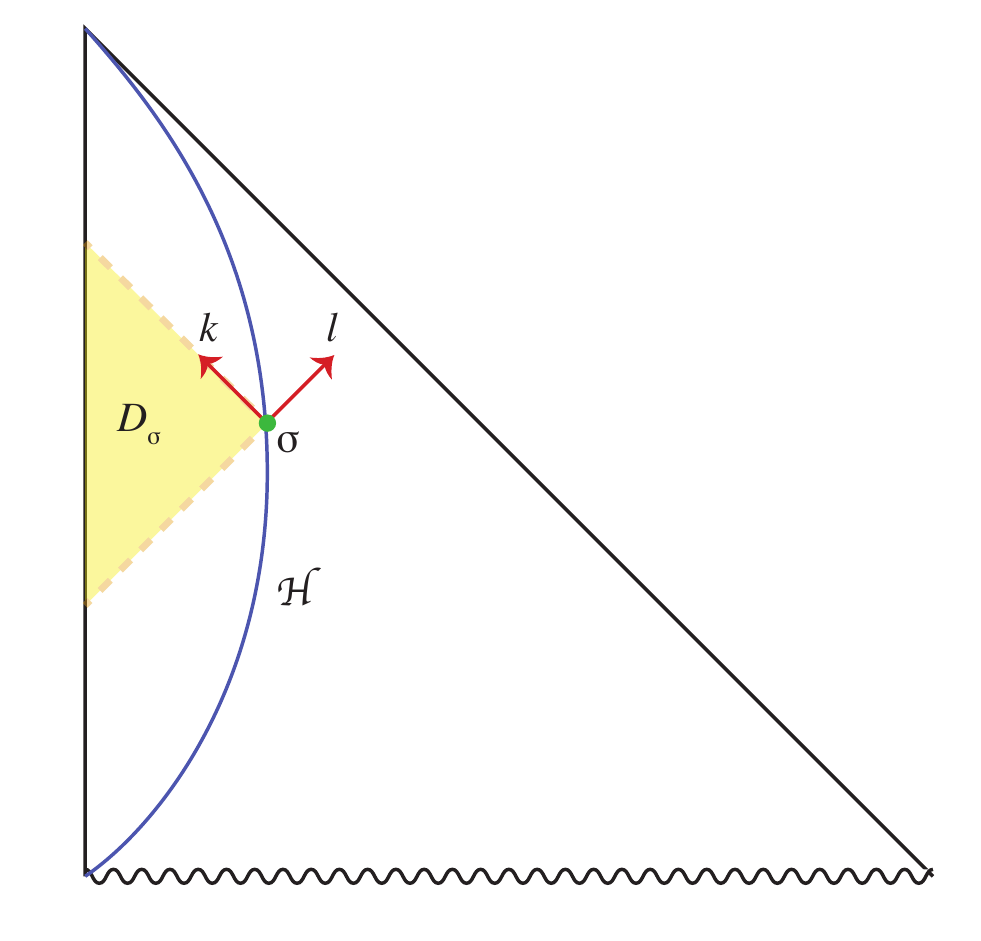}
\caption{An  example of a past holographic screen $\mathcal{H}$.  One particular leaf $\sigma$ is highlighted here along with its null orthogonal vector fields $k$ and $l$ satisfying $\theta^k=0$ and $\theta^l>0$.  The causal region $D_\sigma$ plays a critical role in our generalization of holographic entanglement entropy.}
\label{fig_screen_def}
\end{figure}

Some of the simplest examples of holographic screens arise in the ``observer-centered'' case where we take $\{N_r\}$ to be the set of past light-cones of an observer's worldline in some spacetime.  In the case of FRW cosmology with the observer taken to be comoving, such holographic screens are just apparent horizons.  Figure \ref{fig_screen_def} shows an example of such a holographic screen.  See also figures \ref{dS_identifications} (top) and \ref{fig_big_crunch}.   Because past holographic screens are often generated in this way, we will mostly focus on the case of past screens throughout this paper.  However, all results below apply equally well to future screens with appropriate modifications.

Because null foliations are highly non-unique, holographic screens are also non-unique.  For example, in the observer-centered case, a past holographic screen can be obtained obtained by considering the surfaces of maximal area on the past light-cones of an observer's worldline if the maximal area surfaces are anti-trapped and compact which we assume.  In this case, performing a modification to the worldline will modify the holographic screen.\footnote{Note that the non-uniqueness of holographic screens for a given spacetime fits well with the ideas of \cite{Nomura:2013nya, Nomura:2014voa, Nomura:2014woa} where a strong emphasis is placed on the importance of ``fixing the gauge'' in quantum gravity.  This was clearly discussed in \cite{Nomura:2013nya} in which the role of a gauge-fixed apparent horizon (essentially a holographic screen though not a past or future screen) was discussed.  We do not commit to the pictures described in these papers.}    From this point of view, holographic screens appear to be ``pro-complementarity'' objects. The potential importance of this aspect of screens is further discussed below.

Suppose that $\mathcal{H}$ is a past holographic screen.  Let $\sigma$ be a leaf of $\mathcal{H}$ and let $k$ and $l$ denote, respectively, the ingoing and outgoing future-directed null surface-orthogonal vector fields on $\sigma$.  (It may be useful to refer to figure \ref{fig_screen_def}.)   Then, the condition that $\sigma$ is marginally anti-trapped means that
\begin{equation} 
\label{theta_cond}
\begin{split}
\theta^k = 0 \\
 \theta^l > 0
 \end{split}
\end{equation}
where $\theta^k$ and $\theta^l$ denote the expansion of congruences in the $k$ and $l$ directions at $\sigma$.

Every holographic screen comes with a \emph{fibration}.  A fibration is a family of curves generated by a nonvanishing vector field $h$ on and tangent to $\mathcal{H}$ with the property that $h$ is orthogonal to every leaf.  If we extend the vector fields $k$ and $l$ to all of $\mathcal{H}$ (so that they are surface orthogonal to every leaf), then $h = \alpha l + \beta k$ where $\alpha$ and $\beta$ are scalar functions on $\mathcal{H}$.  $h$ is not required to be timelike, spacelike or null and, in fact, can switch between these three cases on one screen.  Thus, holographic screens need not have definite signature.\footnote{This is the key distinguishing feature between past (and future) holographic screens and related objects including future outer trapping horizons and  dynamical horizons \cite{Hayward:1993wb,Hayward:1997jp,Ashtekar:2003hk,Ashtekar:2005ez} that were introduced in an attempt to find a ``quasi-local'' definition of a black hole.  Past and future holographic screens can be regarded as a synthesis such ideas with those of \cite{Bousso:1999xy}.}  Lacking a definite signature, normalization of $h$ is arbitrary.  Nonetheless, it is convenient to write the leaves of $\mathcal{H}$ as $\sigma(r)$ where $r$ is some (non-unique) parameter and to then normalize $h$ by the condition $dr(h) = 1$.

Bousso and Engelhardt proved that $\alpha >0$ at every point in $\mathcal{H}$ and concluded that leaves have strictly increasing area \cite{Bousso:2015mqa,Bousso:2015qqa}.  More precisely, the area of $\sigma(r_2)$ is greater than the area of $\sigma(r_1)$ if $r_2 >r_1$.  In fact, if $\norm{\cdot}$ denotes the area functional, then
\[
\frac{d}{dr} \norm{\sigma(r)} = \int_{\sigma(r)} d^{d-2}y \sqrt{g^{(\sigma(r))}}\:  \alpha \: \theta^l
\]
which is positive by equation \ref{theta_cond} and the fact that $\alpha>0$.  Here, $g^{(\sigma(r))}$ denotes the induced metric on $\sigma(r)$.  Note, in particular, that
the area is \emph{strictly} increasing for all intervals of $r$. The inequality would not be strict if it were not for the genericity conditions of  \cite{Bousso:2015qqa}.

\subsection*{Definition and Properties of Holographic Screen Entanglement Entropy}

As before, let $\mathcal{H}$ denote a past holographic screen.  Everything below can be modified to the case of a future holographic screen without difficulty.

It is helpful to emphasize the following result which follows from the genericity conditions of  \cite{Bousso:2015qqa}:

\begin{itemize}
\item  \emph{Strict Focusing.}  If $B$ is a codimension 2 spacelike surface, the four surface-orthogonal null congruences have strictly decreasing expansion as they move away from $B$.
\end{itemize}
This means that there is always enough matter content everywhere in the spacetime to focus neighboring null geodesics.  If $M$ fails to satisfy this condition, it can be made to do so by sprinkling a very small amount of classical matter everywhere.  

As discussed above, there is a unique foliation of $\mathcal{H}$ into anti-trapped leaves.  Let $\sigma$ be a particular leaf in this foliation and let $k$ and $l$ denote the vector fields on $\sigma$ that satisfy equation \ref{theta_cond}.  Because $M$ is globally hyperbolic, there exists a Cauchy surface $S_0$ containing $\sigma$ such that  $S_0 \setminus \sigma$ consists of a disconnected interior and exterior.  The interior of $S_0$ is defined so that a vector on $\sigma$ pointing toward the interior takes the form $c_1 k - c_2 l$ with $c_1,c_2>0$.  Let $S$ denote the union of the interior of $S_0$ with $\sigma$.  We will assume that $S$ is compact and that it has the topology of a solid ball.  Now let $D_\sigma$ be the domain of dependence of $S$, $D_\sigma = D(S)$, with the convention that $D_\sigma$ includes orthogonal null surfaces generated by $k$ and $-l$.  

\begin{figure}
\centering
\includegraphics[width=8cm]{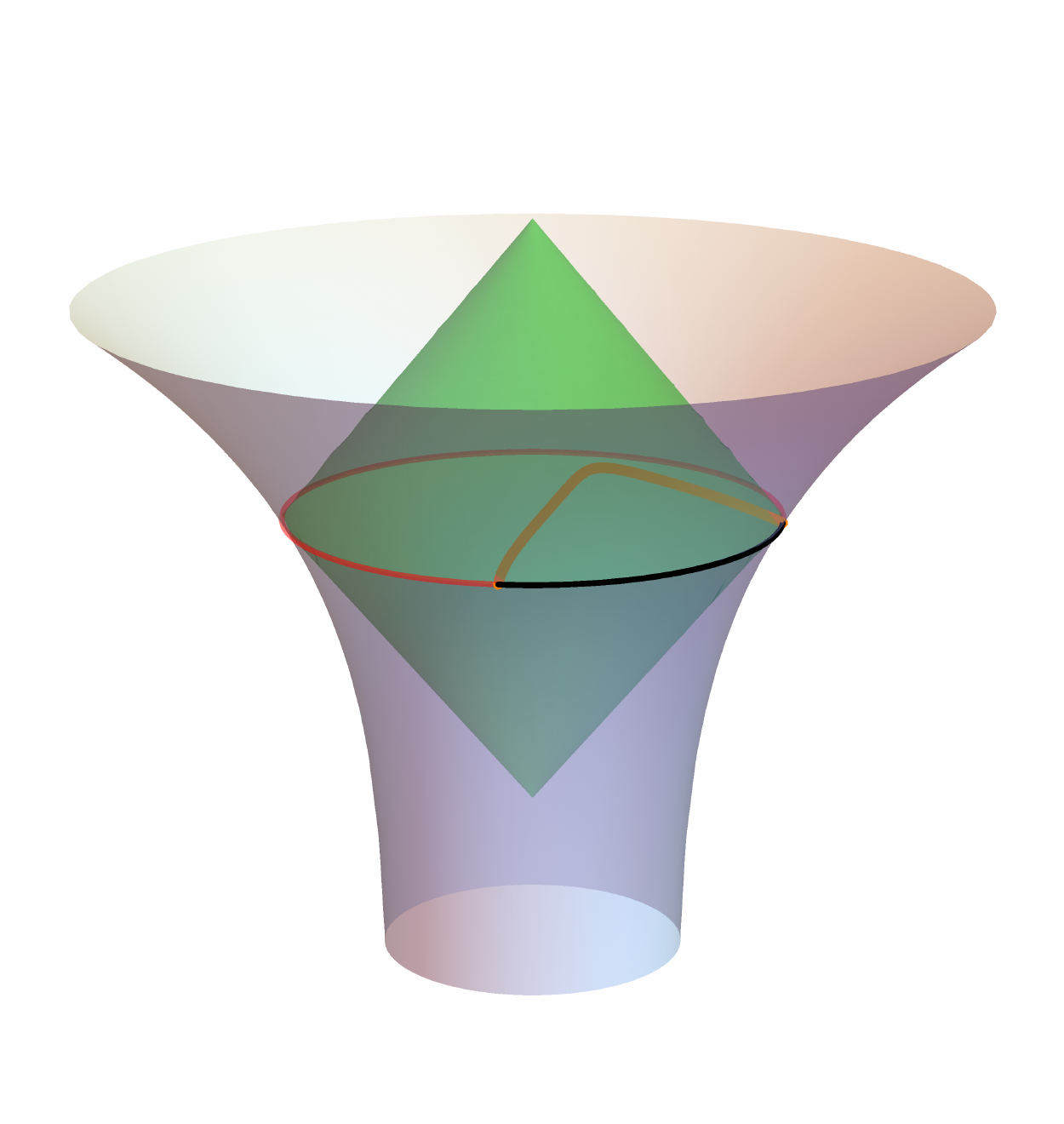}
\caption{This figure depicts our construction of holographic entanglement entropy in general spacetimes.  The horn-shaped surface is a past holographic screen $\mathcal{H}$.  The black and red codimension 2 regions together form a single leaf $\sigma$.  The black segment represents a region $A$ and the extremal surface $\ext(A)$ (orange) is anchored to its boundary.  The causal region $D_\sigma$ is the green diamond (both interior and boundary).  Note that $\ext A \subset D_\sigma$.}
\label{fig_screen_3d}
\end{figure}

Suppose that $A$ is a $d-2$ dimensional submanifold of $\sigma$ with a boundary.  Consider the set of extremal codimension 2 surfaces that are anchored to and terminating at $\partial A$, and contained entirely in $D_\sigma$ (see figure \ref{fig_screen_3d}).  In section \ref{sec_SSA} we will give conditions on $D_\sigma$ that ensure that this set is not empty.  Taking the existence of such a surface for granted, let the one of minimal area be denoted by $\ext(A)$ and define the \emph{holographic screen entanglement entropy} (or \emph{screen entanglement entropy} for brevity) of $A$ as

\begin{equation}
\label{entropydef}
S(A) = \frac{\norm{\ext(A)}}{4}.
\end{equation}

The quantity $S(A)$ is the most natural generalization of the HRT proposal to general spacetimes.  We emphasize that we have defined screen entanglement entropy geometrically without reference to a quantum theory.  The term ``entanglement entropy'' is only meant suggestively.  Nonetheless, below we state a \emph{screen entanglement conjecture}: that $S(A)$ is in fact the von Neumann entropy of a subsystem of a holographic quantum state for general spacetimes.  Regardless of the validity of this conjecture, we are free to study $S(A)$ as we have defined it.  As we will see, the properties of holographic screens ensure that screen entanglement entropy possesses numerous properties reminiscent of von Neumann entropy which we now discuss. 

\subsubsection*{Properties of Holographic Screen \\* Entanglement Entropy and Extremal Surfaces}

\begin{itemize}
\item \emph{Existence and Containment.} In section \ref{sec_SSA} we provide conditions for $\ext (A)$ to exist.  This is a nontrivial issue because of the ``containment condition'' that $\ext A \subset D_\sigma$.  Arguments that $D_\sigma$ contains an extremal surface rely critically on the assumption that $A$ is in a leaf of a holographic screen.  Moreover, the condition that  $\ext (A) \subset D_\sigma$ gives rise to properties of holographic screen entanglement entropy like strong subadditivity (see below) and will allow us to reasonably define an entanglement wedge for $A$.  For an example of the importance of the containment condition, see equation \ref{not_contained} below and the paragraphs around it.

\item \emph{(Strong) Subadditivity.} Suppose that $A$ and $B$ are regions in $\sigma$.  Then,
\[
S(A) + S(B) \geq S(A \cup B) + S(A \cap B)
\]
where $S$ is the function defined in \ref{entropydef}.  This result holds regardless of whether or not $A$ and $B$ intersect as long as we take the convention that $S(\emptyset) = 0$.  As we will see in section \ref{sec_SSA}, the proof of this is a modified version of Wall's \cite{Wall:2012uf} ``maximim'' proof for the HRT case.  This does not mean that strong subadditivity is an obvious result: most of the work in section \ref{sec_SSA} is to show that the properties of leaves of holographic screens are sufficient to generalize Wall's arguments to our context.

\item \emph{Page Bounded.}  Define the extensive entropy of $A$ as $S_{\textrm{extensive}}(A) = \norm{A} / 4$.  Then, the holographic screen entanglement entropy satisfies the following \emph{Page bound}:\footnote{The term ``Page bound'' is motivated by Page's considerations of the entanglement entropies of subsystems \cite{Page:1993df}.}
\begin{equation}
\label{thermal_bound}
S(A) \leq \min\{S_{\textrm{extensive}}(A) , S_{\textrm{extensive}}(\sigma \setminus A) \}.
\end{equation}
This is a simple consequence of the maximin construction we give in section \ref{sec_SSA}.  Note that the area law for holographic screens implies that this inequality becomes a weaker constraint if we transport $A$ along the fibration vector field defined above.  In certain cases, the inequality saturates and $S(A)$ approaches a ``random entanglement limit.''  (See section \ref{sec_cosmology} for examples of this in cosmology.)

\item \emph{Reduction to the HRT Proposal.} As explained in detail in \cite{Bousso:1999cb}, the AdS boundary can be regarded as a holographic screen.  In this case, surfaces of constant time in the dual field theory correspond to leaves, and our proposal becomes identical to the covariant holographic entanglement entropy conjecture of \cite{Hubeny:2007xt}.

\end{itemize}

\subsection*{The Screen Entanglement Conjecture}

We are now in a position to state our conjecture about the role of $S(A)$ in quantum gravity.
This conjecture is the primary concern of this paper.  Nonetheless, we emphasize that the mathematical developments below (e.g. the proof that $S(A)$ satisfies standard properties of von Neumann entropy) do not rely on any conjectural statements.  

Our proposal can be regarded as an extension of a covariant holographic principle due to Bousso which we now review very briefly.  In \cite{Bousso:1999cb}, Bousso integrated the ideas of \cite{'tHooft:1993gx, Susskind:1994vu} with his covariant entropy conjecture \cite{Bousso:1999xy} and proposed that each marginal surface $B$ foliating a holographic screen is associated with a Hilbert space $\mathcal{H}_B$ of dimension $\exp(\textrm{area}(B)/4)$ and that states in $\mathcal{H}_B$ holographically define the state on a null surface $N$ passing through $B$ in the marginal direction.  For our purposes, this holographic principle takes the following form. To each leaf $\sigma$ of a past or future holographic screen we assign a density matrix $\rho_\sigma$.  The density matrix acts on a Hilbert space of dimension $\exp (\textrm{area}(\sigma)/4)$ which may be a subspace of a ``complete'' Hilbert space.\footnote{The concept that the states corresponding to any particular approximately fixed geometry form a subspace of a complete Hilbert space is due to Nomura \cite{Nomura:2011dt, Nomura:2011rb}.  In his formulation, a larger Hilbert space for arbitrary geometries is a direct sum over subspaces for each geometry.  This direct sum itself is only a subspace of the complete Hilbert space which may include an ``intrinsically stringy'' subspace with no geometrical interpretation.  This construction may provide insight into how quantum mechanics can be unitary despite the fact that screens have non-constant area.}  The covariant entropy bound suggests that $\rho_\sigma$ encodes the quantum information on the null slice generated by $k$ and $-k$ where $k$ is the null vector field with $\theta^k=0$ on $\sigma$.

We now assume Bousso's holographic principle and state our new conjecture.  We propose that every region $A$ of $\sigma$ (up to string scale resolution) corresponds to a subsystem of the Hilbert space that $\rho_\sigma$ acts on.  We conjecture that the von Neumann entropy of that subsystem in the density matrix $\rho_\sigma$ is given, at leading order, by $S(A)$ as we have defined it in equation \ref{entropydef}. 

We refer to this statement as the \emph{screen entanglement conjecture}.  Because a holographic quantum theory dual to arbitrary spacetimes is not known, the screen entanglement conjecture is not  a mathematical statement about the relation between two known theories (as in the case of HRT).  Instead, our conjecture suggests a way to compute properties of quantum states in an unknown theory.  It is our hope that this will, in fact, be a step toward developing a quantum theory for arbitrary spacetimes.

\subsubsection*{Nonuniqueness of Holographic Screens and Frame-Dependence in Quantum Gravity}

It was emphasized above that in a given spacetime, there is no unique preferred holographic screen.  As a consequence, screen entanglement entropy
cannot even be defined before first deciding on a particular choice of a screen.  This might seem to put the screen entanglement conjecture
on haphazard footing, but we explain here why this arbitrariness is, in fact, a necessary feature of any generalization of holographic entanglement entropy to general spacetimes.  

Conventional holographic entanglement entropy in AdS/CFT is reference frame dependent in the following sense.  Consider an observer in an asymptotically AdS spacetime $M$ with conformal boundary $\partial M$ following a worldline $p(\tau)$.  Here, $\tau$ is the proper time parameter of the observer.  At a given value of $\tau$, we can consider a spacelike \emph{cut} of
the boundary   \cite{Newman:1976gc,Engelhardt:2016wgb}: 
\[
C(\tau) = \partial J_-\big(p(\tau)\big)  \cap \partial M.
\]
Here, $J_-(q)$ denotes the causal past of a point $q$.  A region $A_\Omega$ on $C(\tau)$ can be specified by considering a portion $\Omega$ of a small sphere on the tip of the past light cone of  the point $p(\tau)$ and following points in $\Omega$ down null geodesics until $\partial M$ is reached.  Thus, once the trajectory $p(\tau)$ is decided upon, we can use the HRT formula to compute $S(A_\Omega,\tau)$, the holographic entanglement entropy of the region $A_\Omega$ on the cut $C(\tau)$.  If the trajectory is changed, $S(A_\Omega,\tau)$ correspondingly transforms.  
At the level of the dual CFT, this discussion corresponds to the fact that  quantum states and their time-dependence have a gauge-redundancy that is fixed by making a choice of time-slicing on the boundary $\partial M$.

In the case of the screen entanglement conjecture and a spacetime that is not asymptotically AdS, a null foliation must be selected to fix a holographic screen.  As discussed above, a simple way to do this is to choose a curve $p(\tau)$, and, at any given $\tau$, follow along the past light-cone of $p(\tau)$ until a marginal surface $\sigma(\tau)$ is obtained.  The role of $\sigma(\tau)$ is analogous to that of the cut $C(\tau)$ in asymptotically AdS spacetimes.  The foliation dependence of screen entanglement entropy
is closely tied to the frame-dependence of the HRT formula.  This is an example of ``fixing the gauge'' in quantum gravity, a concept developed in  \cite{Nomura:2011rb}.

In the case of asymptotically AdS spacetimes, no matter what worldline $p(\tau)$ is chosen, the union of all of its cuts will always be a subset of the boundary.\footnote{This follows trivially from the definition of $C(\tau)$.  However, in some cases past directed null geodesics of $p(\tau)$  may fail to reach $\partial M$.} In general spacetimes however, the particular holographic screen obtained by taking the union over all $\tau$ of $\sigma(\tau)$ will  depend on the choice of the worldline.  Thus, the surface on which holographic quantum states are defined is no longer tethered to the spacetime.  This is a basic property of Bousso's holographic principle, not one that arises only in the more extended framework of this paper.  We regard this aspect of holographic screens as being in the spirit of black hole complementarity, where quantum information is not attached to a fixed spacetime position (e.g. a qubit is not inside or outside a black hole) until an observer is selected to describe the system.

\section{Proofs of Strong Subadditivity and Other Relations}
\label{sec_SSA}

In this section we prove key technical results about holographic screen entanglement entropy including many of the properties advertised above.  The notation and conventions we will use are the same as those given in section \ref{section2}.  In particular, $\mathcal{H}$ is a past holographic screen in a globally hyperbolic spacetime of dimension $d$ that satisfies the genericity conditions of \cite{Bousso:2015qqa}.  $\sigma$ is a compact leaf of $\mathcal{H}$ which we assume to have the topology of $S^{d-2}$.  $k$ and $l$ are null orthogonal vector fields on $\sigma$ satisfying equation \ref{theta_cond}.  $S_0$ is a Cauchy slice containing $\sigma$ and $S$ is the portion of $S_0$ that is enclosed by $\sigma$  including $\sigma$ itself (the enclosed side is defined in section \ref{section2}).  $S$ is assumed to have the topology of a compact $d-1$ ball.  $D_\sigma$ is the domain of dependence of $S$.

As always, the case of a future holographic screen is omitted because it presents no additional subtlety.

\subsection*{Existence and Containment of Extremal Surfaces}

As discussed in section \ref{section2}, it is nontrivial and critical to show the existence of  an extremal surface anchored to $\partial A$  that lies entirely in $D_\sigma$.  We now prove that such a surface exists under very generic conditions.  Our first step is to show that $\ext(A)$ exists in the case that $D_\sigma$ is compact.  This is a common situation\footnote{Suppose that the future and past ingoing light-sheets of $\sigma$ terminate at caustics rather than singularities.  Let $C_+$ and $C_-$ denote the set of the first caustics encountered (local or nonlocal) by null geodesics in the future and past light sheets respectively.  Then, if  $D_\sigma = J_-(C_+) \cap J_+(C_-)$, we can conclude that $D_\sigma$ is compact.  This follows from the fact that $C_\pm$ inherits the compactness of $\sigma$ and from the fact that global hyperbolicity implies that $J_-(K_1) \cap J_+(K_2)$ is compact if $K_1$ and $K_2$ are compact.} although it is not the case if the ingoing light sheets of $\sigma$ encounter a singularity.

\begin{figure}
\centering
\includegraphics[width=8cm]{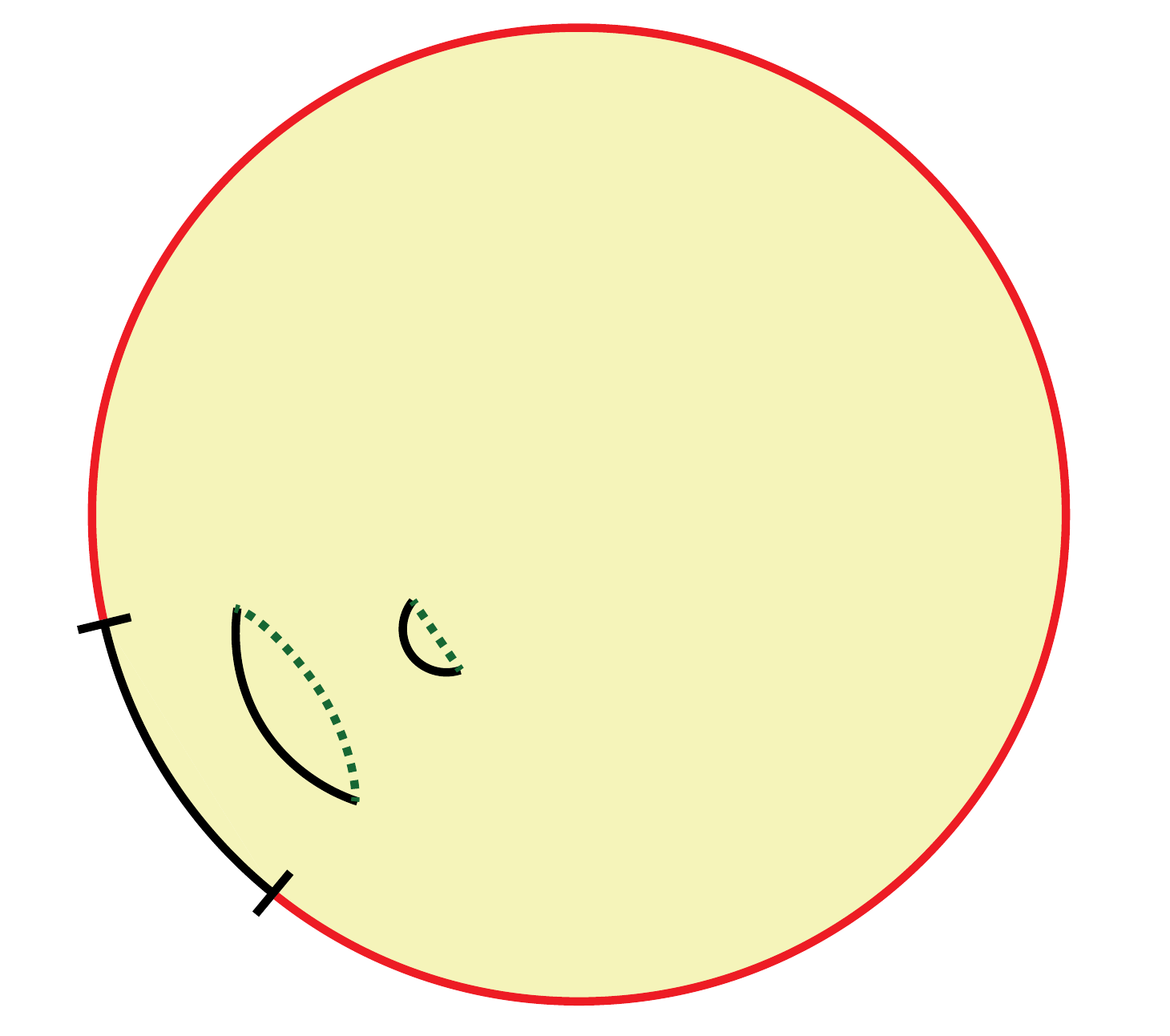}
\caption{The proof of lemma \ref{inside_lem} involves a continuous family of surfaces $A_s$ along with their extremal surfaces (dotted curves).}
\label{fig_lemma_1}
\end{figure}

\begin{lemma}
\label{inside_lem}
If $D_\sigma$ is compact, then there exists a codimenson 2 extremal surface anchored and terminating at $\partial A$ that lies entirely in $D_\sigma$  and that intersects $ \partial D_\sigma$ only at  $\partial A$. 
\end{lemma}

\begin{proof}
Let $\Sigma_+$ and $\Sigma_-$ denote the future and past ingoing light-sheets of $\sigma$.  We now extend $\Sigma_-$ to a slightly larger light-sheet, $\tilde{\Sigma}_-$, by following the future directed null congruence of $l$.  Because $\theta^l>0$ on $\sigma$, we can make this extension so that $\tilde{\Sigma}_-$ has $\theta^l>0$ everywhere and so that there exists an open set in $\tilde{\Sigma}_-$ containing $\sigma$.

In the language of \cite{Engelhardt:2013tra}, both $\Sigma_+ \setminus \sigma$ and $\tilde{\Sigma}_-$ are extremal surface barriers because they have negative expansion in the $k$ and $-l$ directions respectively.  Moreover, $\partial D_\sigma \subset  \left( \Sigma_+ \setminus \sigma \right)  \cup \tilde{\Sigma}_-$.  It follows that $\partial D_\sigma$ is itself an extremal surface barrier for extremal surfaces in the interior\footnote{In \cite{Engelhardt:2013tra}, extremal surfaces are confined to regions referred to as the ``exterior'' of an extremal surface barrier.  The interior of $D_\sigma$, i.e. $D_\sigma \setminus \partial D_\sigma$, is analogous to exterior regions studied by Wall and Engelhardt.} of $D_\sigma$.

Now consider the region $A$.  The spherical topology\footnote{We remind the reader that our conventions are those laid out in the first paragraph of section \ref{section2}.  In particular, we are making simplifying topological assumptions about $\sigma$ and $S$.  We will leave it to future work to investigate the consequences of relaxing these assumptions.} of $\sigma$ ensures that it is possible to introduce a continuous one-parameter family of submanifolds of $D_\sigma$, $A_s$, such that

\begin{itemize}
\item $A_0$ consists of a single point in the interior of $D_\sigma$
\item $A_1 = A$
\item for $0<s<1$, $A_s$ is a codimension 2 submanifold of the interior of $D_\sigma$ that is diffeomorphic to $A$.
\end{itemize}
This is shown in figure \ref{fig_lemma_1}.  Note, in particular, that if $s<1$, $A_s \cap \partial D_\sigma = \emptyset$.  

If $\epsilon>0$ is sufficiently small, then the extremal surface of minimal area that is anchored to $\partial A_\epsilon$ lies entirely in the interior of $D_\sigma$.  Denote this extremal surface by $\Gamma(\epsilon)$. Consider increasing the value of the parameter $s$ from $\epsilon$ to $1$.  For each value of $s$, construct an extremal surface $\Gamma(s)$ (not necessarily the one of minimal area) anchored to $\partial A_s$. The compactness of $D_\sigma$ (which ensures that it is bounded and has no singularities) together with the fact that, as discussed above, $\partial D_\sigma$ is an extremal surface barrier, allows us to take $\Gamma(s)$ to not jump discontinuously  and to be contained in the interior of $D_\sigma$ for all $s<1$.  When we take the limit sending $s$ to 1, the extremal surface anchored to $\partial A$ must intersect $\partial D_\sigma$ at $\partial A$ and nowhere else:  if it did intersect $\partial D_\sigma$ outside of $\partial A$, the extremal surface would be locally tangent to an extremal surface barrier with strictly nonzero null extrinsic curvature.

\end{proof}

\begin{figure}
\centering
\includegraphics[width=8cm]{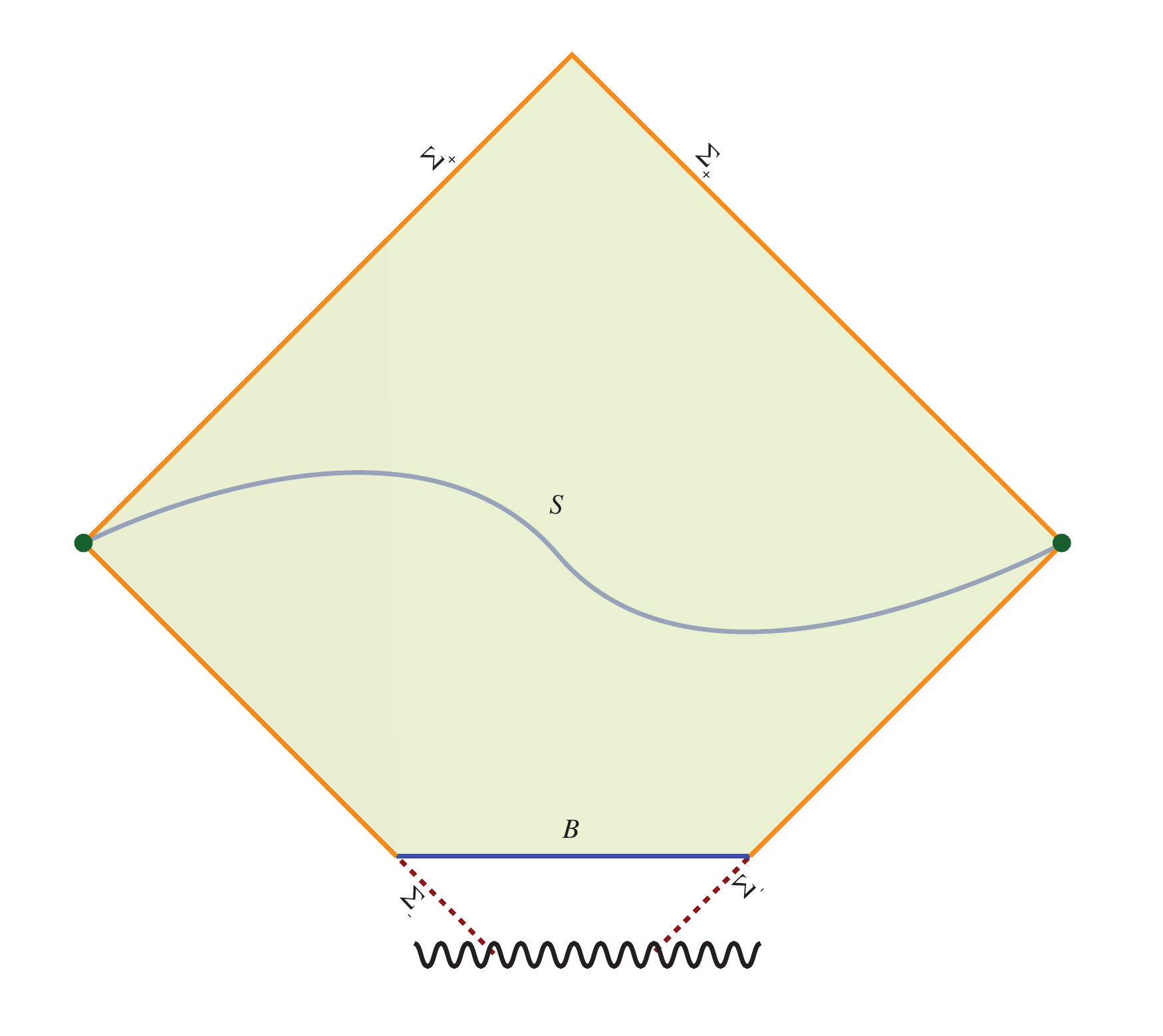}
\caption{The idea of a compact restriction is shown here.  The restriction $R$ is the shaded region along with its boundary, the blue and orange lines.  $\partial R$ consists of two parts: an extremal surface barrier $B$ (blue) and a portion of $\partial D_\sigma$ (orange).  In this figure, the barrier $B$ protects extremal surfaces in $R$ from a singularity.  Not shown are extremal surfaces in $R$, none of which contact $\partial R$ except at their anchor on the leaf $\sigma$.}
\label{fig_compact_restriction}
\end{figure}

The unwanted assumption that $D_\sigma$ is compact (which fails in the event that $\Sigma_+$ or $\Sigma_-$ encounter a singularity) can be dropped if there exists a codimension 0 submanifold (with boundary) of $D_\sigma$, $R$, which ``restricts'' extremal surfaces (see figure \ref{fig_compact_restriction}).  By this we mean that

\begin{enumerate}
\item $R$ is compact,
\item There exists an open set $U$ containing $S$ with $D_\sigma \cap U = R \cap U$, and
\item $\partial R = \left( \partial D_\sigma \cap R \right) \cup B$ where $B$ is an extremal surface barrier for codimension 2 extremal surfaces inside in $R$.
\end{enumerate}

These conditions for $R$ are designed to ensure that $R$ can be used in lemma \ref{inside_lem} in place of $D_\sigma$ without difficulty.  The existence of such regions $R$ relies on the existence of the barrier $B$.  The arguments in theorem 11 of \cite{Wall:2012uf} show that Kasner singularities are always protected by such barriers.  Hartman and Maldacena \cite{Hartman:2013qma} encountered a barrier protecting black hole singularities from codimension 2 extremal surfaces.  Constant time slices in FRW spacetimes are another example of suitable barriers.\footnote{Many extremal surfaces are anchored at singularities and thus pass through barriers.  This is irrelevant because the barriers we are discussing here play the of $\partial D_\sigma$ in the proof of lemma \ref{inside_lem}.  As a region $A_s$ is deformed from a point inside $R$ into $A \subset \sigma$, extremal surfaces anchored to $\partial A_s$ cannot smoothly pass $B$ or $\partial D_\sigma$.}

Any region $R \subset D_\sigma$ satisfying the conditions will be called a \emph{compact restriction} of $D_\sigma$.  Note that, in particular, if $D_\sigma$ is compact then $D_\sigma$ is a compact restriction of itself.  Our findings can now be summarized by the following improvement upon lemma \ref{inside_lem}:

\begin{theorem}
\label{inside}
If $D_\sigma$ possesses a compact restriction, then there exists a codimenson 2 extremal surface anchored and terminating at $\partial A$ that lies entirely in $D_\sigma$  and that intersects $ \partial D_\sigma$ only at  $\partial A$. 
\end{theorem}

%%%Make it clear that theorem 1 works regardless of the choice of A!!

To better appreciate this theorem, it is helpful show that the statement is false if $\sigma$ is not a leaf of a holographic screen.  Consider $2+1$ dimensional Minkowski space with inertial coordinates $(t,x,y)$ and let $\mathcal{C}$ denote the large cylinder satisfying $x^2+y^2 = R^2$ with $R \gg 1$.  Consider the two line segments on $\mathcal{C}$ that are approximately given by 
\begin{align}
\label{not_contained}
\begin{split}
A &= \{(t = \frac{1}{2}|x|,-1>x \geq 0, y = R)\}  \\
B &= \{ (t = \frac{1}{2}|x|,0 \leq x < 1,y=R) \}
\end{split}
\end{align}
and construct any spacelike ``time slice'' on $\mathcal{C}$, $\sigma$, that includes $AB$.  It is easy to see that the extremal surface anchored to $\partial \left( AB \right)$ is a straight line that is timelike related to $AB$ and thus fails to lie within the domain of dependence $D_\sigma$.  To see how severe this problem is, note that the segments $A$ and $B$ fail to satisfy subadditivity of entanglement entropy.  That is, the inequality $S_A + S_B \geq S_{AB}$ is false.  Note that in this example $\sigma$ fails to satisfy equation \ref{theta_cond} because of the kink at $A \cap B$.

\subsection*{A Maximin Construction for Holographic Screens}

Theorem \ref{inside} ensures that holographic screen entanglement entropy is a well-defined quantity in a broad set of cases. We will now demonstrate that this quantity satisfies expected properties of entanglement entropy.  To do this, it is very useful to closely follow \cite{Wall:2012uf} and introduce a maximin construction of $\ext A$.  Our construction will be slightly modified from that used for HRT surfaces anchored to the AdS boundary.  Wall's maximin prescription involves considering a collection of Cauchy slices that are anchored only to $\partial A$.  Because we already know that $\ext A$ lies inside of $D_\sigma$, we will introduce a stronger constraint requiring that we only consider achronal slices that are anchored to all of $\sigma$.

\subsubsection*{Definition and Existence of \Mm(A)}

Our setup remains unchanged.  Fix a past (or future) holographic screen $\mathcal{H}$ in a globally hyperbolic spacetime and let $\sigma$ be a leaf.  We take a Cauchy surface $S_0$ containing $\sigma$ and define $S$ as the closure of the portion of $S_0$ inside of $\sigma$.  As before, we require that $S$ is compact and that it has the topology of a solid $d-1$ ball.  Let $D_\sigma = D(S)$.  We also fix a region $A$ in $\sigma$ with a boundary.  Now define $\mathcal{C}_\sigma$ as the collection of codimension 1 compact achronal surfaces that are anchored to $\sigma$ and that have domain of dependence $D_\sigma$.  Note, in particular, that $S \in \mathcal{C}_\sigma$.  Moreover, note that the global hyperbolicity of $D_\sigma$ ensures that every element of $\mathcal{C}_\sigma$ has the same topology as $S$: that of a compact $d-1$ ball.

Take any $\Sigma \in \mathcal{C_\sigma}$.  Let $\min(\partial A, \Sigma)$ denote the codimension 2 surface of minimal area\footnote{\label{footnote_homology}  Wall \cite{Wall:2012uf} added the condition that $\min(\partial A, \Sigma)$ be homologous to $A$.  While this condition ought to be included in our discussion as well, the assumption that $S$ (and thus every element of $\mathcal{C}_\sigma$) has the topology of a compact $d-1$ ball makes a homology condition trivial.  We leave the task of investigating more general topologies to future work.} on $\Sigma$ that is anchored to $\partial A$.  The existence of $\min(\partial A, \Sigma)$ is guaranteed by the compactness of $\Sigma$ and theorem 9 of \cite{Wall:2012uf}.  Define a function $F: \mathcal{C_\sigma} \rightarrow [0,\textrm{area}(A)]$ by $F(\Sigma) = \textrm{area}(\min(\partial A, \Sigma))$.  Now assume that there exists a $\Sigma_0$ in $\mathcal{C}_\sigma$ that maximizes $F$ (globally).  We now define $\min  (\partial A, \Sigma_0)$ as the maximin surface of $A$, and we will denote it by $\Mm (A)$.  If there are several maximin surfaces, $\Mm(A)$ can refer to any of them.

The existence of $\Mm(A)$ can be proven in many cases by appropriately importing the arguments of theorems 10 and 11 in \cite{Wall:2012uf} which we only briefly describe here.  Consider the Cauchy surface $S_0$ which can be identified as a slice in a foliation of Cauchy surfaces $\{S_t \}$.  Using this definition of time, we can identify a surface $\Sigma \in \mathcal{C}_\sigma$ with a function $t_\Sigma : S_0 \rightarrow \bold{R}$ in a natural way: if $I_x$ denotes the integral curve of $\partial_t$ that passes through a point $x \in S$, then $\Sigma = \{ I_x \cap S_{t_\Sigma(x)}  | x \in S\}$.  From this viewpoint, $F$ can be regarded as a real-valued functional on $\{ t_\Sigma \}$.  Now if $D_\sigma$ is compact, we can find the maximum and minimum values of $t$ for the set $D_\sigma$ to obtain an upper and lower bound on $t_\Sigma$ that applies for all $\Sigma$.  Moreover, the condition that $\Sigma$ be compact and achronal ensures that $\{ t_\Sigma \}$ is equicontinuous.  These facts imply that $\mathcal{C}_\sigma$ is compact (with the uniform topology) and that the extreme value theorem applies to the function $F$.

In the case where $D_\sigma$ is not compact (for instance, due to a singularity terminating a light sheet of $\sigma$), we can still argue that $F$ has a maximum as long as $D_\sigma$ satisfies a condition similar to but slightly stronger than the ``compact restriction'' idea discussed above.  Suppose that $B_+$ is a surface in $\mathcal{C}_\sigma$ which is identical to $\Sigma_+$ in some neighborhood of $S$.  For any $\Sigma \in \mathcal{C}_\sigma$, define another surface $\bar{\Sigma}$ by $t_{\bar{\Sigma}} = \min \{ t_\Sigma , t_{B_+} \}$.  If $B_+$ has the property that for any $\Sigma$ we have $F(\Sigma) \leq F(\bar{\Sigma})$, then we will say that $B_+$ is a \emph{future maximin barrier}.  A past maximin barrier is defined analogously as a surface $B_- \in \mathcal{C}_\sigma$, identical to $\Sigma_-$ in a neighborhood of $S$, such that for any $\Sigma$ we have $F(\Sigma) \leq F(\bar{\Sigma})$ where $\bar{\Sigma}$ is defined by $t_{\bar{\Sigma}} = \max \{ t_\Sigma , t_{B_-} \}$.

Now if $D_\sigma$ possesses both a past and future maximin barrier, then we can restrict our attention to the subset of surfaces in $\mathcal{C}_\sigma$ that satisfy $t_{B_-}\leq t_\Sigma \leq t_{B_+}$.  Let $\mathcal{C}_\sigma(B_-,B_+)$ denote this restricted set.  Because $B_-$ and $B_+$ are compact, $J_+(B_-) \cap J_-(B_+)$ is compact and so the set $\mathcal{C}_\sigma(B_-,B_+)$ is compact in the uniform topology and $F$ has a maximum $\Sigma_0 \in \mathcal{C}_\sigma(B_-,B_+)$.  The definition of past and future maximin barriers ensures us that if $\Sigma \in \mathcal{C}_\sigma$, then $F(\Sigma_0) \geq F(\Sigma)$.  Thus, $\Sigma_0$ is a global maximum for $F$ and we can safely define $\min (\partial A, \Sigma_0)$ as the maximin surface of $A$, $\Mm (A)$.

As in the case of the compact restriction of $D_\sigma$ used in theorem \ref{inside}, it is difficult to find examples where $D_\sigma$ does not possess a past and future barrier.  Wall \cite{Wall:2012uf} argued that such barriers protect maximin surfaces from a wide range of singularities: approximately Kasner singularities, BKL singularities, and FRW big bangs all lead to past or future maximin barriers.  If $\Sigma_\pm$ simply terminate at caustics rather than singularities, then $B_\pm = \Sigma_\pm$ are barriers.  In any event, if $B_\pm$ exist, then the region $J_+(B_-) \cap J_-(B_+)$ provides a compact restriction of $D_\sigma$ in the sense of theorem \ref{inside}.  Thus, the existence of $B_\pm$ ensures both the existence of $\Mm(A)$ as well as the existence of $\ext (A)$.  From here on, we will simply take for granted that a past and future maximin barrier exist.

\subsubsection*{Equivalence of \Mm (A) and \ext(A)}

Below we will argue that $\Mm (A) = \ext (A)$.  However, it is first very useful to introduce two additional definitions first.

\begin{enumerate}
\item Take $\Sigma \in \mathcal{C}_\sigma$ and let $\Gamma$ be a codimension 2 surface anchored to $\partial A$ that lies in $D_\sigma$.  Consider the intersection between $\Sigma$ and the future and past-directed orthogonal null surfaces of $\Gamma$ that are directed toward $A$.  This intersection is called the representative of $\Gamma$ on $\Sigma$ and will be denoted by $\rep (\Gamma, \Sigma)$.

\item The domain of dependence of codimension 1 achronal surfaces anchored to $A \cup \ext A$ lying in $D_\sigma$ will be called the entanglement wedge of $A$.
\end{enumerate}

Note that  $\rep (\Gamma, \Sigma)$ is itself a codimension 2 surface anchored to $\partial A$ that lies on $\Sigma$.  Moreover, if $\Gamma$ is extremal then, by the focusing theorem, $\textrm{area} ( \rep (\Gamma, \Sigma)) \leq \textrm{area}(\Gamma)$.

We now demonstrate that our maximin procedure always finds $\ext A$, the extremal surface of minimal area that is anchored to $\partial A$ and which lies in $D_\sigma$.  While much of the proof is similar to the arguments in \cite{Wall:2012uf}, we will have to pay special attention to the possibility that the maximin surface could run into the boundary of $D_\sigma$. 

\begin{theorem}
\label{mmext}
$\Mm (A) = \ext(A)$.
\end{theorem}
\begin{proof}
The argument of theorem 15 in \cite{Wall:2012uf} immediately shows that if a point $p \in \Mm (A)$ is also in the interior of $D_\sigma$ (i.e. $D_\sigma \setminus \partial D_\sigma$), then $\Mm(A)$ is extremal at $p$.  In particular, if $\Mm(A) \cap \partial D_\sigma = \partial A$, then $\Mm(A)$ is an extremal surface everywhere.  We now argue that $\Mm (A)$ in fact cannot ever intersect $\partial D_\sigma$ outside of $\partial A$. 

Suppose there exists $p \in \Mm(A) \cap (\partial D_\sigma \setminus \partial A)$. There must be an open neighborhood of $p$ in $\Mm(A)$ (open in the d-2 dimensional manifold $\Mm(A)$) that is entirely contained in $\partial D_\sigma$.  If this were not the case, $\Mm(A)$ would be extremal at points arbitrarily close to $p$ and would thus be extremal at $p$.  Moreover, $\Mm(A)$ would be tangent to $\partial D_\sigma$ at $p$.  However, $\partial D_\sigma$ is an extremal surface barrier (see lemma \ref{inside_lem}) so this is not possible.  There are now two cases to consider.

\begin{figure}
\centering
\includegraphics[width=8cm]{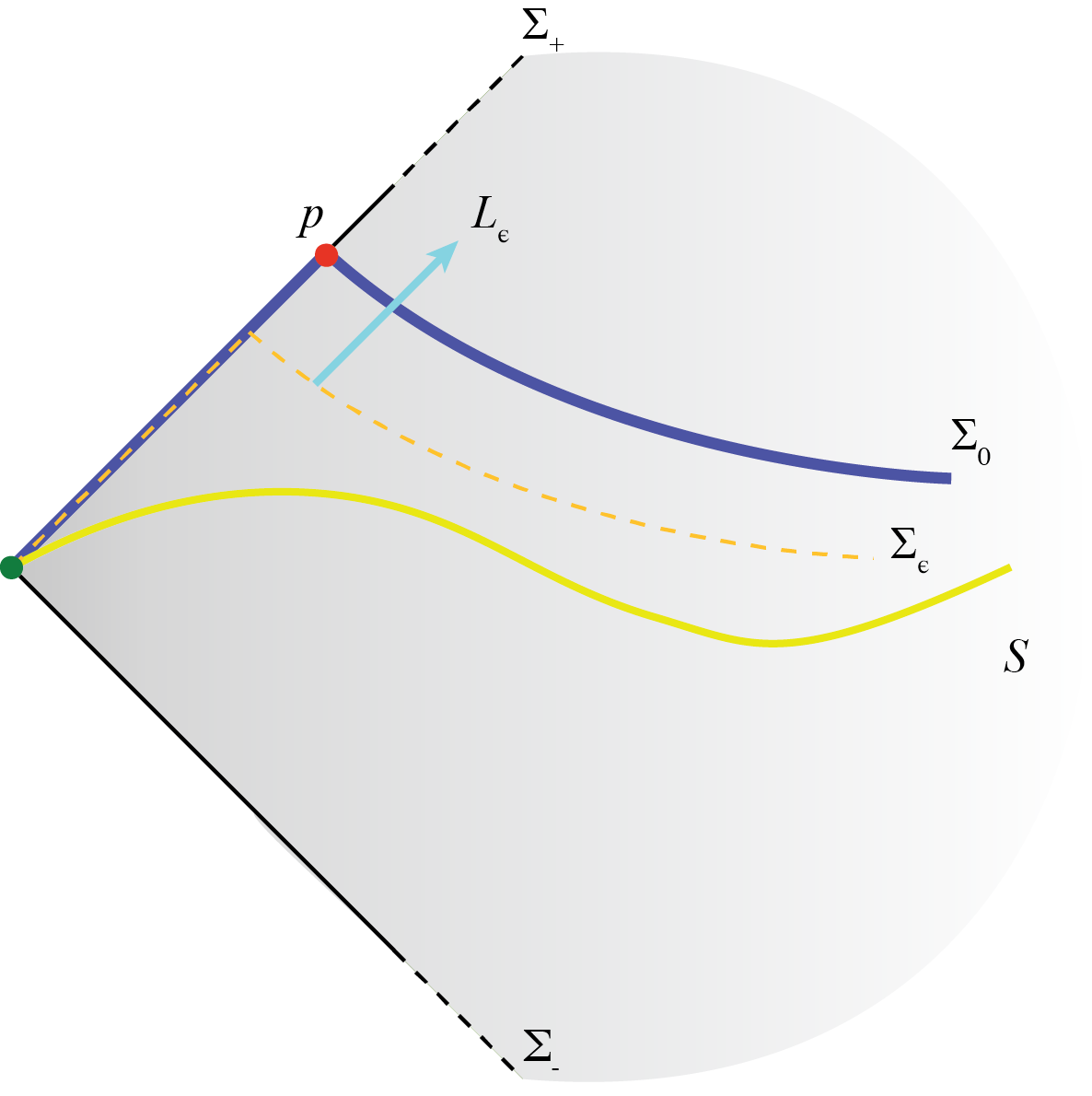}
\caption{This figure depicts the argument of case 1 of the proof of theorem \ref{mmext}.  Note that the surface $S$ is shown here for reference and that it does not play a critical role in the proof.  The shaded region is $D_\sigma = D(S)$ and the green dot is (a cross-section of) the leaf $\sigma$.}
\label{fig_Mm_proof_case_1}
\end{figure}

\begin{itemize}
\item \emph{Case 1}: $p \in \partial D_\sigma \setminus \sigma$. \\
Figure \ref{fig_Mm_proof_case_1} illustrates a construction that we will use for this case.  Take $p \in \Sigma_+$ (the case of $p \in \Sigma_-$ is no different).  By construction, $\Mm(A)$ is minimal on a surface $\Sigma_0$.  There exists a (dimension d-1) open subset $U$ of $\Sigma_0$ containing $p$ such that $U \cap \Mm (A) \subset \Sigma_+$.  Moreover, we can require that $U$ is ``split'' by $\Mm (A)$ into two disconnected sets, $N$ and $V$, such that $N$ is the side of $U$ closer to $\sigma$.  Since $\Sigma_0$ is anchored to $\sigma$, we must have that $N \subset \Sigma_+$ and, in particular, $N$ is null.  On the other hand, $V$ cannot be a subset of $\Sigma_+$.  If it were, then $\Mm(A)$ could decrease its area by being deformed up $\Sigma_+$ (by the focusing theorem).  In particular, we can take $U$ small enough to ensure that $V$ is nowhere null in the direction of $k$.

We now consider the process of slightly sliding $\Sigma_0$ down $\Sigma_+$. More precisely, take a small parameter $\epsilon >0$ and a corresponding one-parameter family of slices $\{\Sigma_\epsilon\}$ that are slightly deformed from $\Sigma_0$ in a way we now describe (an example of $\Sigma_\epsilon$ is depicted in figure  \ref{fig_Mm_proof_case_1} by an orange dashed line).  The surface $\Mm (A) \cap U$ is described by a function $\lambda_0(x)$ giving the affine distance from $\sigma$ up to $\Mm(A)$ at a point $x \in \sigma$.  Now put $\lambda_\epsilon(x) = \lambda_0(x) - \epsilon f(x)$.  Here, $f:\sigma \rightarrow [0,1]$ is a smooth weighting function which equals 1 at the null generator $x_p$ that $p$ lies on. We take $f$ to go to zero smoothly as $x$ moves away from $x_p$, equaling zero exactly when $x$ corresponds to a point outside of $U\cap \Mm(A)$.  For $\lambda < \lambda_\epsilon(x)$, we require that the surface $\Sigma_\epsilon$ is identical to $\Sigma_+$. We extend $\Sigma_\epsilon$ beyond $\lambda_\epsilon$ by parallel transporting tangent vectors on $\Mm(A)$ directed toward $V$ down to $\lambda_\epsilon$.  This prescription does not uniquely fix $\Sigma_\epsilon$, but it is sufficient for our purposes.

Consider the one-parameter family of codimension 2 curves $\min(\partial A, \Sigma_\epsilon)$.  For any $\epsilon>0$, let $L_\epsilon$ denote the future-directed null congruence of $\min(\partial A, \Sigma_\epsilon)$ that points toward the interior of $D_\sigma$ (see figure \ref{fig_Mm_proof_case_1}).  The continuity of  $\min(\partial A, \Sigma_\epsilon)$ as $\epsilon$ varies and the fact that $\Sigma_+$ is a light sheet ensures that there exists and $\epsilon_0>0$ such that for $\epsilon<\epsilon_0$,
\begin{itemize}
\item $L_\epsilon$ intersects $\Sigma_0$ to form a codimension 2 surface on $\Sigma_0$ anchored to $\partial A$ and
\item $L_\epsilon$ has negative future-directed expansion in the region between $\min(\partial A, \Sigma_\epsilon)$ and its intersection with $\Sigma_0$.
\end{itemize}
Denote this intersection by $C_\epsilon$ and observe that $C_0 = \Mm(A)$.  But $\Mm(A)$ is minimal on $\Sigma_0$ so for sufficiently small $\epsilon$,
\[
\textrm{area}(\Mm(A)) < \textrm{area}(C_\epsilon) \leq \textrm{area}(\min(\partial A, \Sigma_\epsilon))
\]
which contradicts the assumption that $\Mm(A)$ has area greater than or equal to the minimal area surface on any slice.  Note that the last inequality above follows from the focusing theorem applied to $L_\epsilon$.

\item \emph{case 2}: $p \in \sigma$. \\
Assume that there exists a (dimension d-2) open subset of $\Mm (A)$ that is contained in $\sigma$. (If not, there must be such an open set in $\partial D_\sigma \setminus \sigma$ which just leads to case 1 above.)  Now consider the null vector field $k$ on $\sigma$ and the geodesics generated by it.  Follow these geodesics from $\sigma$ up along $\Sigma_+$ by a short affine distance $\epsilon>0$ to generate a new codimension 2 surface, $\sigma_\epsilon$, which limits to $\sigma$ when $\epsilon \rightarrow 0$.  The focusing theorem now gives rise to a modified version of equation \ref{theta_cond} at $\sigma_\epsilon$: 

\begin{align}
\label{theta_cond_ep}
\begin{split}
\theta_\epsilon^l > 0 \\
 \theta_\epsilon^k < 0.
 \end{split}
\end{align}
Along with moving $\sigma$ up the light-sheet, we also translate $A$ up the sheet to a one-parameter family of surfaces $A_\epsilon$ that limit to $A$.   Consider the maximin construction applied to codimension 2 surfaces anchored to $\partial A_\epsilon$ that lie on codimension 1 surfaces anchored to $\sigma_\epsilon$.  We denote the result by $\Mm(A_\epsilon)$.  We also define $D_{\sigma_\epsilon}$ in the obvious way.  Now this maximin procedure leads to the same two cases that we are now studying.  The first case, where $\Mm(A_\epsilon)$ intersects $\partial D_{\sigma_\epsilon} \setminus \sigma_\epsilon$ proceeds exactly as it did with $\epsilon=0$.  Now suppose that $ \Mm(A_\epsilon)$ has an open set contained in $\sigma_\epsilon$.  $ \Mm(A_\epsilon)$ must be minimal on some slice $\Sigma_\epsilon$.  However, equation \ref{theta_cond_ep} implies that $\sigma_\epsilon$ has negative (inward) extrinsic curvature on $\Sigma_\epsilon$.  It is thus impossible for $\Mm(A_\epsilon)$ to be minimal on $\Sigma_\epsilon$ since its area could be decreased by ``cutting corners.''

We can thus conclude that $\Mm(A_\epsilon) \cap \partial D_{\sigma_\epsilon} = \partial A_\epsilon$.  This implies that $\Mm(A_\epsilon)$ is extremal.  Taking the limit as $\epsilon \rightarrow 0$, we conclude that $\Mm(A)$ is extremal.  But, given our assumption that part of $\Mm(A)$ lies on $\sigma$, equation \ref{theta_cond} shows that $\Mm(A)$ cannot be extremal since extremal surfaces have zero null expansion in all directions. 
\end{itemize}

At this point it is proven that $\Mm(A)$ is extremal.  All that is left is to show that, of all the extremal surfaces in $D_\sigma$ that are anchored to $\partial A$, $\Mm(A)$ is the smallest.  Let $\Sigma_0 \in \mathcal{C}_\sigma$ be a slice on which $\Mm(A)$ is minimal.  If $\Gamma$ is another extremal surface anchored to $\partial A$ then, as a result of the focusing theorem, we find that 
\[
\textrm{area}(\Mm(A)) \leq  \textrm{area} ( \rep (\Gamma, \Sigma_0)) \leq \textrm{area}(\Gamma).  
\]
\end{proof}

We are now in a position to prove a variety of properties of screen entanglement entropy.  We begin with the ``Page bound'' advertised in section \ref{section2}.
\begin{corollary}
\label{thermal_bound_thm}
If $A$ is a region in the leaf $\sigma$, then
\begin{equation*}
S(A) \leq \min\{S_{\textrm{extensive}}(A) , S_{\textrm{extensive}}(\sigma \setminus A) \}
\end{equation*}
where $S$ deonotes the holographic screen entanglement entropy of $A$ and $S_\textrm{extensive}(X)$ denotes the area of a region $X \subset \sigma$ divided by 4.
\end{corollary}
\begin{proof}
$S(A) = \textrm{area(\Mm (A))}/4$ but $\Mm (A) = \min (\partial A, \Sigma_0)$ for some $\Sigma_0 \in \mathcal{C}_\sigma$.  Both $A$ and $\sigma \setminus A$ are codimension $d-2$ dimensional surfaces on $\Sigma_0$ anchored to $\partial A$ so the area of $\Mm (A)$ is less than or equal to the areas of both $A$ and $\sigma \setminus A$. 
\end{proof}

Next we turn to the proof of strong subadditivity for holographic screen entanglement entropy (other properties of entanglement entropy that admit covariant geometrical bulk proofs can be imported here as well).  Unlike the case of theorems \ref{inside} and \ref{mmext}, the arguments below are essentially identical to those of \cite{Wall:2012uf} with little additional subtlety.  We start with our version of theorem 17 in \cite{Wall:2012uf} which states that if $B \subset A$, then $\ext A$ lies ``outside'' of $\ext B$.

\begin{theorem}
\label{same_surface}
Suppose that $A$ and $B$ are regions in the leaf $\sigma$ with $B \subset A$.  Then, 
\begin{enumerate}
\item the entanglement wedge of $A$ contains the entanglement wedge of $B$,
\item there exists a surface in $\mathcal{C}_\sigma$ on which both $\ext A$ and $\ext B$ are minimal.
\end{enumerate}
\end{theorem}
\begin{proof}[Sketch of Proof:]
The proof is the same as that of theorem 17 of \cite{Wall:2012uf} so we only sketch it here. For any surface in $\Sigma \in \mathcal{C}_\sigma$, consider a pair of codimension 2 surfaces constrained to lie on $\Sigma$, $\Gamma_A$ and $\Gamma_B$, such that $\Gamma_A$ is anchored to $\partial A$ and $\Gamma_B$ is anchored to $\partial B$.  Then let $Z= \textrm{area}(\Gamma_A) + \textrm{area}(\Gamma_B)$.  We now minimize the value of $Z$ by varying over all possible choices of $\Gamma_A$ and $\Gamma_B$.  After that, we maximize the minimal values of $Z$ by varying over all possible $\Sigma$.

This new maximin procedure gives a well-defined answer for the maximinimal value of $Z$.  Moreover, a slice $\Sigma_0$ results on which both $\Gamma_A$ and $\Gamma_B$ are minimal.  On this slice, it is impossible for  $\Gamma_A$ to cross $\Gamma_B$ as this would necessarily give rise to a surface on $\Sigma_0$ anchored to $\partial A$ with smaller area than $\Gamma_A$.   A further observation is that if a connected component of $A$ is distinct from a component of $B$, the corresponding connected components of $\Gamma_A$ and $\Gamma_B$ cannot come into contact even tangentially.  The argument for this is that the component of $\Gamma_B$ would necessarily have a different trace of its spatial extrinsic curvature than $\Gamma_A$ at points close to the contact point.  This would mean that either $\Gamma_A$ or $\Gamma_B$ is not minimal on $\Sigma_0$.

At this point it is known that components of $\Gamma_A$ or $\Gamma_B$ that are distinct have neighborhoods in $\Sigma_0$ that do not intersect the other surface.  Within such neighborhoods, small deviations $\Sigma_0$ and the minimal surfaces can be made that prove that such surfaces are extremal.

The only remaining step is to show that, in fact, $\Gamma_A$ and $\Gamma_B$ are the extremal surfaces in $D_\sigma$ of minimal area.  If $\Gamma^\prime_A$ is an extremal surface in $D_\sigma$ anchored to $\partial A$, then its representation on $\Sigma_0$ must have larger area than that of $\Gamma_A$ but smaller area than that of $\Gamma^\prime_A$.  Thus, $\Gamma_A = \ext A$.  Similarly, $\Gamma_B = \ext B$.  By construction, both are minimal on the same surface $\Sigma_0 \in \mathcal{C}_\sigma$.  Moreover, because $\Sigma_0$ is achronal, we must have that the entanglement wedge of $A$ contains that of $B$.  
\end{proof}

\begin{corollary}
\label{ssa}
Suppose that $A$, $B$, and $C$ are nonintersecting regions in $\sigma$.  Then,
\[
S(AB) + S(BC) \geq S(ABC) + S(B)
\]
where $XY$ denotes $X \cup Y$ and where the function $S$ is defined in equation \ref{entropydef}.
\end{corollary}
\begin{proof}
By theorem \ref{same_surface}, we can find a surface $\Sigma_0 \in \mathcal{C}_\sigma$ such that $\ext B$ and $\ext ABC$ are both minimal on $\Sigma_0$.  Let $\tilde{S}(AB)$ and $\tilde{S}(BC)$ denote the areas of the representations of $\ext AB$ and $\ext BC$ on $\Sigma_0$.  Then,
\[
S(AB) + S(BC) \geq \tilde{S}(AB) + \tilde{S}(BC) \geq S(ABC) + S(B)
\]
where the first inequality follows from the focusing theorem and the second inequality follows from the standard geometric proof of strong subadditivity \cite{Headrick:2007km}.
\end{proof}

Note that the inequality $S(A) + S(B) \geq S(AB)$ follows as a special case of this result.

\section{Extremal Surfaces in FRW Cosmology}
\label{sec_cosmology}

The conventional holographic entanglement entropy prescription, with its limitation to asymptotically locally AdS spacetimes, provides very little information about entanglement structure in cosmology.  One of the most intriguing applications of our proposal, therefore, is to calculate holographic screen entanglement entropy in FRW universes.  Assuming the screen entanglement conjecture, the calculations below give the entanglement entropy of subsystems in quantum states that are dual to cosmological spacetimes.

\subsection*{Holographic Screens in FRW Cosmology}

First we review the holographic screen structure of FRW spacetimes.  Consider a homogeneous and isotropic spacetime with the metric
\begin{equation}
\label{FRW_metric}
ds^2 = -d \tau^2 +  a(\tau)^2 \left( d \chi^2 + f(\chi)^2 d \Omega_2^2 \right)
\end{equation}

where $f(\chi) = \sinh(\chi), \chi,$ or $\sin(\chi)$ in the open, flat, and closed cases respectively.  Before computing extremal surfaces we must decide upon a null foliation for the spacetime and then identify the corresponding holographic screen.  Null foliations (and thus holographic screens) are highly nonunique.  The foliation we will consider here is that of past light cones from a worldline at $\chi=0$.

To find the holographic screen for this foliation, it is convenient to introduce a conformal time coordinate $\eta$ such that $d\tau / d\eta = a$.  Then, the past light cone of the point $(\eta = \eta_0, \chi=0)$ satisfies $\chi = \eta_0 - \eta$.  Spheres along the past light cone can be parameterized by the coordinate $\eta$, and their area is given by 
\begin{equation}
\label{area_eta}
\mathcal{A}(\eta) = 4 \pi \: a\Big(\tau(\eta_0-\eta)\Big)^2 \: f(\eta_0-\eta)^2
\end{equation}
Assuming that $a = 0$ is not merely a coordinate singularity, the condition that $\mathcal{A}$ is maximized is equivalent to the condition that $d \mathcal{A}/d \eta = 0$. Thus, equation \ref{area_eta} gives the condition that fixes the holographic screen:
\begin{equation}
\label{screen_constraint}
\frac{f(\chi)}{f^\prime(\chi)} - \frac{1}{\dot{a}(\tau)}=0.
\end{equation} 
The codimension 1 surface defined by this constraint may be timelike, spacelike, or null, depending on the particular choice of FRW spacetime.  The foliating leaves of this holographic screen are spheres of constant $\tau$ and comoving radius $\chi$ satisfying equation \ref{screen_constraint}.  The covariant entropy bound implies that each leaf has sufficient area to holographically encode the information on one past light cone from the worldline at $\chi=0$ \cite{Bousso:1999xy,Bousso:1999cb}.

Let $\sigma(\tau)$ be the leaf of the holographic screen at time $\tau$ and let $\rho(\tau)$ denote the energy density in the universe (measured by comoving observers) at time $\tau$.  Then, one can write a simple expression for the area of a leaf of the holographic screen at time $\tau$, valid for any $f$:
\begin{equation}
\label{screen_area}
\textrm{area}(\sigma(\tau)) = \frac{3}{2 \rho(\tau)}.
\end{equation}
In particular, this expression shows that holographic screens grow in area as the universe expands.

\subsection*{Extremal Surfaces in de Sitter Space}
Consider 3+1 dimensional de Sitter space of radius $\alpha$.  This spacetime is $S^3 \times \mathbf{R}$ with the metric
\[
ds^2 = -d T^2  + \alpha^2 \cosh^2 \left( \frac{T}{ \alpha} \right) d \Omega_3^2 
\]
where $d \Omega_3^2$ is the metric on a unit 3-sphere.  Despite the fact that this spacetime has the form of equation \ref{FRW_metric} (with $f(\chi) = \sin \chi$), it is an awkward setting for the consideration of holographic screens:  the null expansion on the past or future light cones of any point in de Sitter space goes to zero only at infinite affine parameter.  This suggests that the appropriate ``boundary'' of de Sitter space is past or future infinity.  Even if we do attempt to anchor extremal surfaces to spheres at infinity, the analysis in section \ref{sec_SSA} fails to apply because of the assumption made there that leaves are compact.

Fortunately these difficulties can be averted completely by considering an FRW spacetime that asymptotically approaches de Sitter space at late times.  Specifically, we will consider a spacetime of the form of equation \ref{FRW_metric} with vacuum energy density $\rho_\Lambda$ and, in addition, some matter content $\rho_\textrm{matter}(\tau)$ with the property the matter content gives rise to a big bang at $\tau=0$ and dilutes completely\footnote{In particular, we are not considering spacetimes with a big crunch in this section.} as $\tau \rightarrow \infty$.

Equation \ref{screen_area} immediately implies that
\begin{equation}
\label{dS_area_limit}
\lim_{\tau \to  \infty} \textrm{area}(\sigma(\tau)) =  \frac{3}{2 \rho_\Lambda}
= \: 4 \pi  \alpha^2
\end{equation}
where $\alpha = \sqrt{3/8 \pi \rho_\Lambda}$.  Because of the big bang singularity, we must have that $\textrm{area}(\sigma_{\tau=0})=0$.  Thus, by the area law for holographic screens \cite{Bousso:2015mqa, Bousso:2015qqa}, we can conclude that the leaves of our screen are spheres that monotonically increase in area, starting with $0$ area at the big bang, and expanding to approach the de Sitter horizon of area $4 \pi \alpha^2$ at late $\tau$.

%
%Regardless of the choice of $f(\chi)$, the Friedmann equations imply that for large $\tau$,
%\begin{equation}
%\label{adot_late}
%\frac{1}{\dot{a}(\tau)} \approx \frac{1}{a(\tau)} \sqrt{\frac{3}{8 \pi \rho_\Lambda}} \rightarrow 0.
%\end{equation}
%If $\chi_\textrm{screen}(\tau)$ denotes the comoving radius of leaves of the past holographic screen as a function\footnote{In this class of FRW spacetimes, the past holographic screen is timelike.} of $\tau$, then equations \ref{screen_constraint} and \ref{adot_late} show that $\lim_{\tau \rightarrow \infty} \chi_\textrm{screen}(\tau) = 0$.  This implies that at late times $f(\chi) \rightarrow \chi$ so that
%
%\[
% \chi_\textrm{screen}(\tau) \rightarrow  \frac{1}{a(\tau)} \sqrt{\frac{3}{8 \pi \rho_\Lambda}}
%\]
% and that, if $\sigma(\tau)$ denotes the leaf of the holographic screen at FRW time $\tau$, then
%
%\begin{equation}
%\label{dS_area_limit}
%\lim_{\tau \to  \infty} \textrm{area}(\sigma(\tau)) =  \lim_{\tau \to  \infty} 4 \pi \: \Big( a(\tau) f(\chi_\textrm{screen}(\tau)) \Big)^2 
%= \: 4 \pi  \alpha^2
%\end{equation}
%%

Now focus on a late time leaf $\sigma(\tau)$.  As discussed in section \ref{sec_SSA}, given a region $A\subset \sigma(\tau)$ with a boundary, we can determine the holographic screen entanglement entropy of $A$, $S(A)$, by considering an extremal surface anchored to and terminating at $\partial A$.  In the notation of section \ref{sec_SSA}, $D_{\sigma(\tau)}$ is compact so theorem \ref{inside} implies that an extremal surface anchored to $\partial A$ exists and lies inside of $D_{\sigma(\tau)}$.

For any time $\tau$, define
\begin{widetext}
\begin{equation}
\label{thermal_entanglement}
S_\textrm{Page}^\tau(A) =
\begin{cases}
\frac{1}{4} \textrm{area}(A) &  \textrm{area}(A) \leq  \frac{1}{2} \textrm{area}(\sigma(\tau)) \\
\frac{1}{4}\left(\textrm{area}(\sigma(\tau)) - \textrm{area}(A)\right) & \textrm{area}(A) > \frac{1}{2} \textrm{area}(\sigma(\tau)).
\end{cases}
\end{equation}
\end{widetext}
We allow this definition to extend to a function $S_\textrm{Page}^\infty(A)$ where $A$ is a region on a 2-sphere of radius $\alpha$.  This $\tau=\infty$ case is defined exactly as in equation \ref{thermal_entanglement} if we take $\textrm{area}(\sigma(\infty)) = 4 \pi \alpha^2$.

Below we will present an argument that if $A\subset \sigma(\tau)$, then  
\begin{equation}
\label{dS_entanglement}
\lim_{\tau \to \infty} S(A) = S_\textrm{Page}^\infty(A).
\end{equation}
(Note that in this limit, it is implied that $A$ is transported to later and later leaves.)  Thus, we will find that as $\tau \to \infty$, $S(A)$ approaches the random entanglement limit discussed in section \ref{section2}.

Any interpretation of this result is necessarily speculative.  Nevertheless, if one assumes the screen entanglement conjecture, then equation \ref{dS_entanglement} implies that the the quantum state of an FRW universe asymptotically approaching de Sitter space has the property that its $O(\alpha^2)$ degrees of freedom are almost randomly entangled with one-another.  At earlier times, the degrees of freedom are not randomly entangled because $S(A) < S_\textrm{Page}^\infty(A)$.

\subsubsection*{Random Entanglement and the Static Sphere Approximation}

\begin{figure}
\centering     %%% not \center
\subfigure{\label{fig:a}\includegraphics[width=70mm]{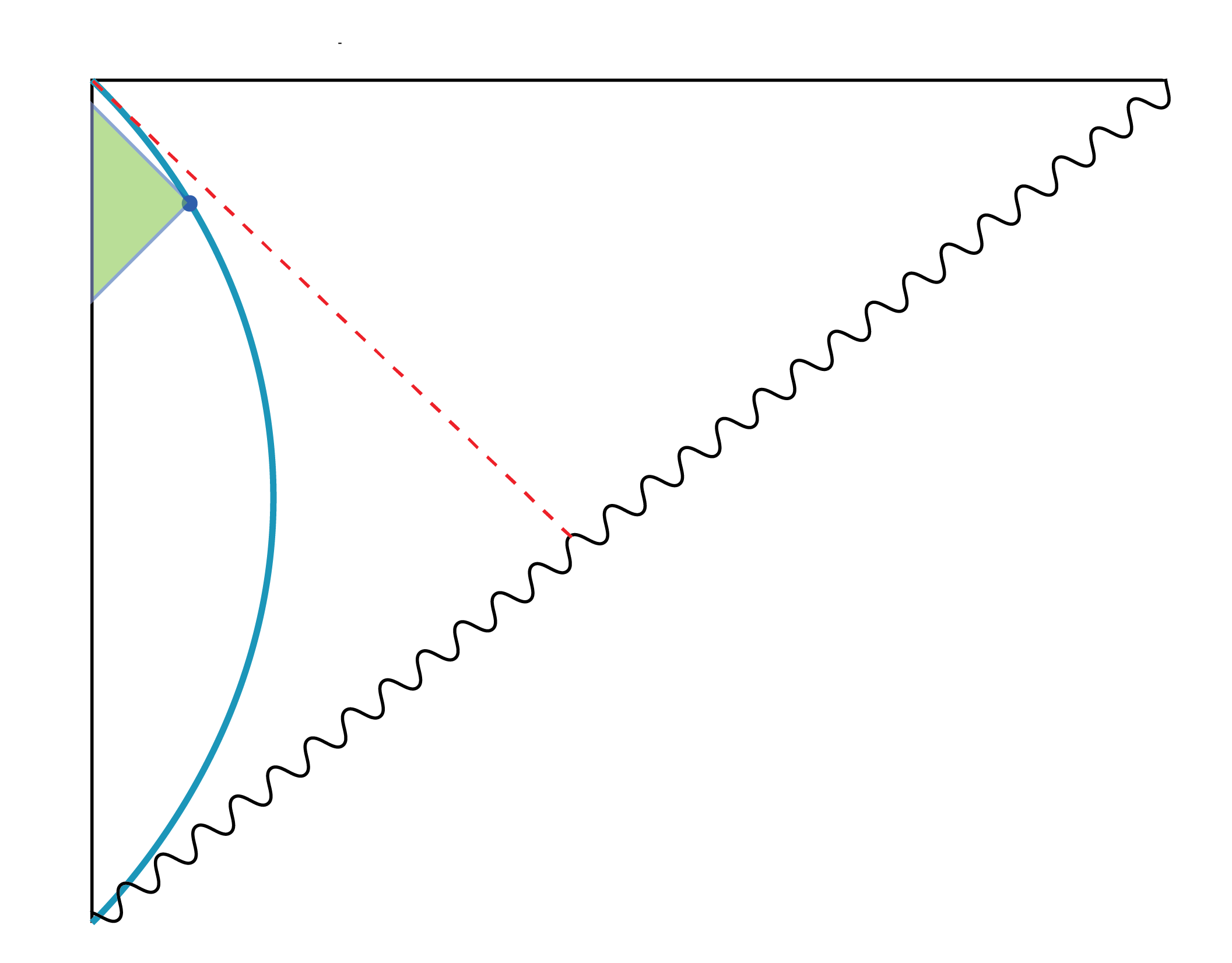}}
\subfigure{\label{fig:b}\includegraphics[width=70mm]{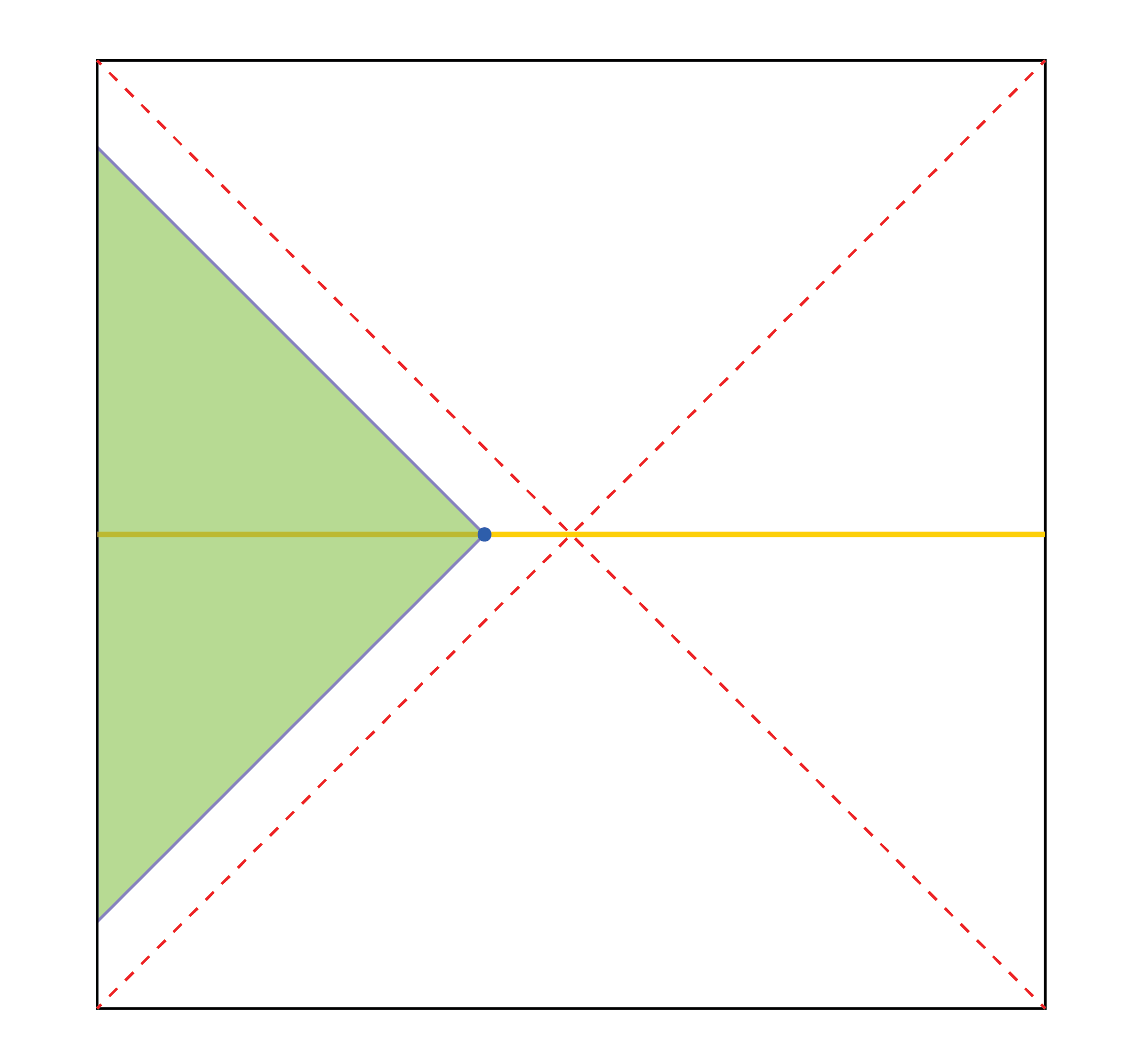}}
\caption{The domain of dependence $D_{\sigma(\tau)}$ for a late time leaf in the flat FRW universe (the small green triangle in the upper diagram) can be approximately mapped to a domain of dependence $D_{\tilde{\sigma}(\tau)}$ in empty de Sitter space (lower diagram).  The mapping becomes increasingly accurate as $\tau$ becomes larger.  The effect of increasing $\tau$ is to move the green triangle in the upper diagram into the top-left corner (along the blue curve), while the green triangle in the lower diagram moves to the right and approaches the entire left static wedge.}
\label{dS_identifications}
\end{figure}

We now present a combination of rigorous arguments, numerical data, and analytic approximations suggesting that the approximate de Sitter cosmological spacetimes discussed above saturate the random entanglement bound in the $\tau \to \infty$ limit.  As before, $\sigma(\tau)$ denotes a leaf at time $\tau$ in an FRW universe with vacuum energy as well as matter energy that dilutes at late time.

The entire region $D_{\sigma(\tau)}$ has a metric that can be made arbitrarily similar to that of a patch of empty de Sitter space by making $\tau$ large.  To see this, first note that points in $D_{\sigma(\tau)}$ have $\chi < \chi_\textrm{screen}(\tau)$ and $\chi_\textrm{screen}(\tau)$ can be made arbitrarily small by making $\tau$ large.  (This follows from equation \ref{dS_area_limit} and the fact that $\lim_{\tau\to \infty} a(\tau) =\infty$.)  Meanwhile, the conformal diagram for our spacetime immediately shows that the minimal value of $\tau$ in $D_{\sigma(\tau)}$ can be made arbitrarily large by making $\tau$ large.  Thus $D_{\sigma(\tau)}$ can be made to only cover arbitrarily large $\tau$ and arbitrarily small $\chi$, in which case our metric of equation \ref{FRW_metric} takes the form
\begin{equation}
\label{flat_dS}
ds^2 \approx -d \tau^2 + c \: e^{2 \tau/\alpha} ( d \chi^2 + \chi^2 d \Omega_2^2 )
\end{equation}
where $c$ is a constant and $\alpha$ is the same constant as before.  Here we have made use of the Friedmann equations.  The right-hand side of this equation is precisely the metric of de Sitter space in flat slicing.  De Sitter space can also be described in static coordinates that make a time-translation Killing vector field manifest:
\begin{equation}
\label{static_dS}
ds^2 \approx -\left(1- \frac{r^2}{\alpha^2} \right)dt^2 + \left(1- \frac{r^2}{\alpha^2} \right)^{-1} dr^2 + r^2 d\Omega_2^2.
\end{equation}
Fortunately, $D_{\sigma(\tau)}$ lies in a region that is well-described by either the flat or static slicing of equations \ref{flat_dS} and \ref{static_dS} respectively.

We can now identify $D_{\sigma(\tau)}$ with a region $D_{\tilde{\sigma}(\tau)} $ where $D_{\tilde{\sigma}(\tau)} $  denotes a corresponding region in exact de Sitter space obtained by finding a sphere $\tilde{\sigma}(\tau)$ in the static patch with area matching that of $\sigma(\tau)$.  While it may seem natural to put $\tilde{\sigma}(\tau)$ at large static time, we can use the $t$ translational symmetry of de Sitter space to place $\tilde{\sigma}(\tau)$ at $t=0$ for all $\tau$.  The effect of increasing $\tau$ is simply to bring $\tilde{\sigma}(\tau)$ closer to the bifurcation sphere on the de Sitter horizon.  This identification is illustrated in figure \ref{dS_identifications}.  Note that as $\tau \to \infty$, the geometry of $D_{\sigma(\tau)} $ and $D_{\tilde{\sigma}(\tau)} $ become arbitrarily similar.

Consider the region $A\subset \sigma(\tau)$ which can be identified with a region $\tilde{A} \subset \tilde{\sigma}(\tau)$.  At large $\tau$, $ \tilde{\sigma}(\tau)$ approaches the equator of a 3-sphere of radius $\alpha$.  The equator itself is an extremal surface so with $\tau<\infty$ but still large, there must be an extremal surface that is close to $\tilde{A}$ but not exactly on it.  Its area will be slightly less than that of $\tilde{A}$.  Note, moreover, that if the area of $\tilde{A}$ exceeds half the area of the equator, then a smaller extremal surface can be obtained by considering the complement of $A$.

This suggests but does not yet prove that at large $\tau$, the holographic screen entanglement entropy of $A$ is almost equal to a fourth of its own area in Planck units if $A$ has less area than half of the de Sitter horizon.  What we have proven so far is that an extremal surface exists with area almost equal to that of $A$ (or $4 \pi \alpha^2 - \textrm{area}(A)$).

What if there is another extremal surface with smaller area than the one we have found?  It is easy to see that this is impossible.  Following the notation in section \ref{sec_SSA}, consider the spacelike surface $\Sigma_0$ that, after mapping to $D_{\tilde{\sigma}(\tau)}$, lies at static time $t=0$, and that and terminates at $\tilde{\sigma}$. ($\Sigma_0$ is most of a hemisphere of the 3-sphere.)  The Riemannian geometry of $S^3$ shows that the surface of minimal area anchored to $\partial A$ is the one we have already found.  If $\Gamma$ is another extremal surface (not necessarily lying on $\Sigma_0$), then its representation on $\Sigma_0$, $\rep(\Gamma,\Sigma_0)$, necessarily has larger area than the extremal surface close to the horizon.  But $\textrm{area}(\Gamma) \geq \textrm{area}(\rep(\Gamma,\Sigma_0))$ so we conclude that $\Gamma$ does not have minimal area.

The arguments above show that the random entanglement limit is saturated at large $\tau$.  Taking $0 \ll \tau< \infty $ and $A \subset \sigma(\tau)$, we now explain a way to obtain a more accurate estimate for $S(A)$ than $S_\textrm{Page}^\tau(A)$.  Calculating $S(A)$ without taking the large $\tau$ limit is more involved than what was done above.  Nonetheless, it is worthwhile to investigate this case to better understand how the Page bound limit is approached.  In particular, it is of interest to understand how the discontinuity of the derivative of $S_\textrm{Page}^\infty$ arises.

\begin{figure}
\centering
\includegraphics[width=8cm]{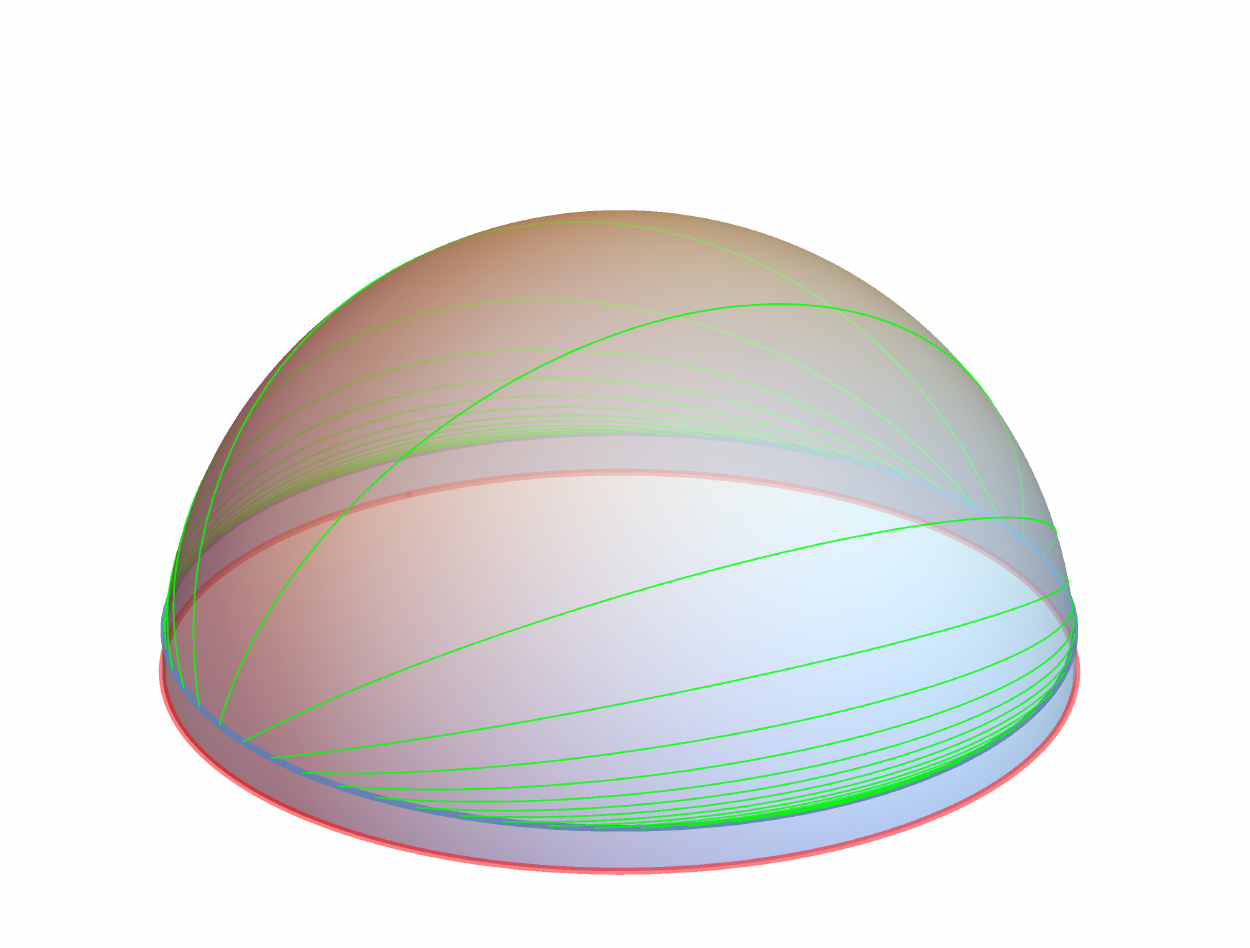}
\caption{The upper hemisphere of a 3-sphere of radius $\alpha$ is half of a static slice in empty de Sitter space and serves as a good approximation for $D_{\sigma(\tau)}$ at large $\tau$.  The blue 2-sphere (appearing as a circle here) lies at constant $z$ (equivalently, constant $r$ where $r$ is the radial coordinate in equation \ref{static_dS}).  This 2-sphere is an approximation for the leaf $\sigma(\tau)$.  Green surfaces depict extremal spherical caps on $S^3$ that approximate $\ext A_\psi$ for various values of $\psi$.  The many samples of extremal surfaces shown here have evenly spaced values of $\psi$.  Figure \ref{fig_CC_dust} provides evidence that this static sphere approximation is accurate at late $\tau$.}
\label{fig_3_sphere}
\end{figure}

We begin by further discussing the role of the 3-sphere in de Sitter space.  Figure \ref{fig_3_sphere} depicts a hemisphere of an $S^3$ of radius $\alpha$ which is precisely half of a static slice of de Sitter space (which we can freely take to be $t=0$).  Define a parameter $z$ as $z = \sqrt{\alpha^2 - r^2}$ where $r$ is the static radius appearing in equation \ref{static_dS}.  Note that a surface of constant $z$ (and static time) is an $S^2$ of area $4 \pi(\alpha^2-z^2)$.  This suggests a way to obtain an approximation for $S(A)$ if $A$ is a region in the leaf $\sigma(\tau)$. Rather than taking $A$ to be a region in $\sigma(\tau)$, we take figure \ref{dS_identifications} seriously and map $A$ to a region in the $S^2$ of constant
\begin{equation}
\label{z_def}
z = \sqrt{\frac{4 \pi \alpha^2 - \textrm{area}(\sigma(\tau))}{4\pi}}
\end{equation}
which ensures that this $S^2$ has the same area as $\sigma(\tau)$.  After this mapping is made, one computes $S(A)$ by finding the extremal surface on the $S^3$ that is anchored to $\partial A$ (which we take to lie at constant $z$).   Below we will refer to this procedure as the ``static sphere approximation.''

Consider regions in $\sigma(\tau)$ that are spherical caps.  Such a cap can be fixed (up to $SO(3)$ rotation) by a zenith opening angle angle $\psi$, so we will denote our region of $\sigma(\tau)$ by $A_\psi$. (With this notation, $A_{\pi/2}$ is a hemisphere and $A_\pi$ is the entire leaf.) The static sphere approximation makes it is clear that for $0< \psi \ll \pi/2$, $\ext A_\psi$ is close to $A_\psi$ itself and that for $\pi/2 \ll \psi < \pi$, $\ext A_\psi$ approaches $\sigma(\tau) \setminus A_\psi$.  As $\psi$ passes the transition angle $\pi/2$, $\ext A_\psi$ quickly passes over the top of the 3-sphere of radius $\alpha$.  The closer $\textrm{area}(\sigma(\tau))$ is to $4 \pi \alpha^2$, the faster $\ext A_\psi$ passes over the top of the sphere.  This explains how the discontinuity in the derivative of $S_\textrm{Page}^\infty(A)$ arises in the large $\tau$ limit.\footnote{For finite $\tau$, there is always another extremal surface on the 3-sphere which goes around the sphere the wrong way.  This surface always has area greater than $\ext A_\psi$ and, in any case, fails to lie in $D_{\sigma(\tau)}$.  However, if we consider the $\tau=\infty$ limit, then $\ext A_\psi$ does not smoothly pass over the hemisphere of the 3-sphere, and in this case, the discontinuity in the derivative of $S_\textrm{Page}^\infty(A)$ is explained by the fact that the surface that wraps around the sphere the ``wrong way'' is now precisely the complement of $A_\psi$ in the equator.  If $\psi$ exceeds $\pi/2$ in this case, then the complement of $A_\psi$ has smaller area than $A_\psi$.  We see that a phase transition occurs only in the exact $\tau=\infty$ limit.}

Because the geometry of $S^3$ is simple, it is not difficult to obtain an explicit (if cumbersome) expression for $S(A_\psi)$ in the static sphere approximation:

%S(A_\psi) \approx   \pi  \left(\alpha ^2 - z^2\right)  \: \sin ^2 \left(\frac{ \cos ^{-1}\left[\frac{z^2}{\alpha^2} +  \left(1-\frac{z^2 }{\alpha^2}\right)  \cos 2 \psi  \right]}{4 \: \sqrt{1-\frac{z^2}{\alpha^2}}}\right)
\begin{widetext}
\begin{equation}
\label{dS_entropy_static}
S(A_\psi) \approx   \pi     \sin ^2 \Bigg(\frac{1}{4} \cos ^{-1}\left[\frac{z^2}{\alpha^2} +  \left(1-\frac{z^2 }{\alpha^2}\right)  \cos 2 \psi  \right] \Bigg)
\end{equation}
\end{widetext}
where $z$ is given by equation \ref{z_def} and, as before, $\alpha = \sqrt{3/ 8 \pi \rho_\Lambda}$.  This expression can be thought of as giving a correction to the ``zeroth order'' expression $S(A_\psi) \approx S_\textrm{Page}^\infty(A_\psi)$.  Taking $\tau<\infty$ will lead to corrections in $1/\tau$ that are not described by the static sphere method.  It is an open question as to whether or not such corrections can, in principle, be of the same (or greater) order in $1/\tau$ as the one we have studied here.  However, numerical data that suggests that the static sphere approximation is accurate at large $\tau$ as we will now see.

\begin{figure}
\centering     %%% not \center
\subfigure[$z/\alpha \approx .05$]{\label{fig:a}\includegraphics[width=85mm]{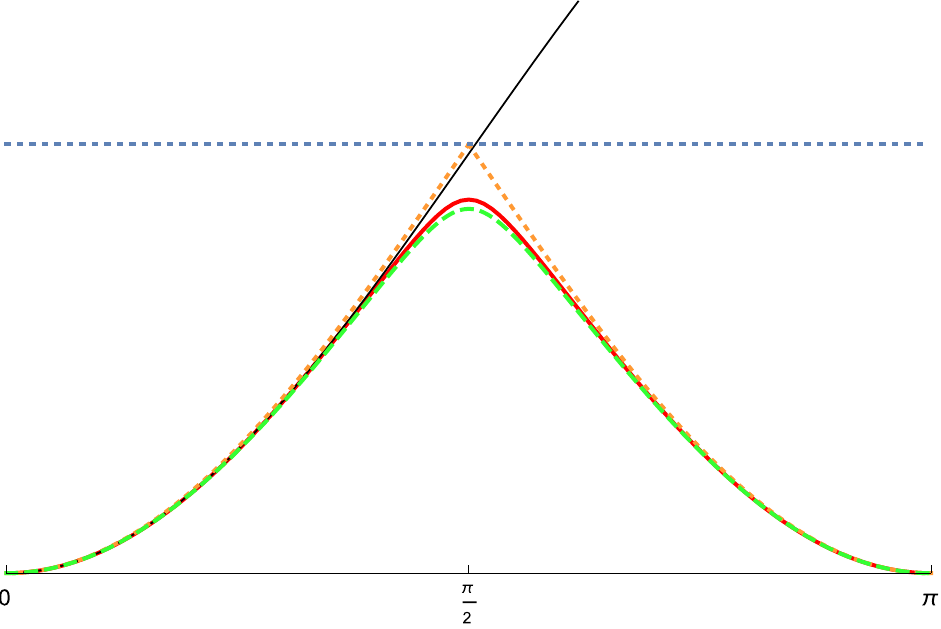}}
\subfigure[$z/\alpha \approx .02$]{\label{fig:b}\includegraphics[width=85mm]{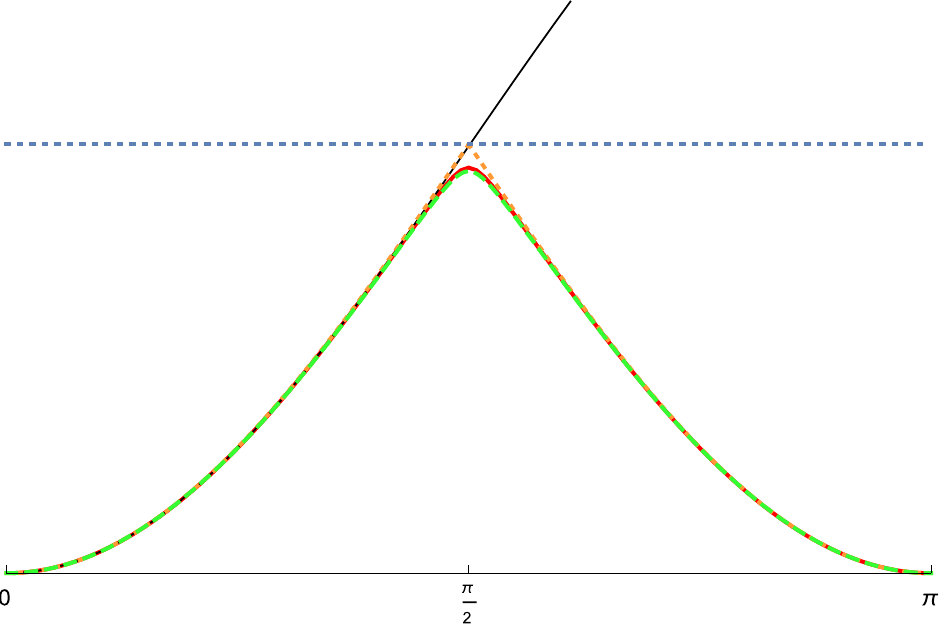}}
\caption{Plots of $S(A_\psi)$ and other quantities for two leaves at different times in a universe with dust and vacuum energy.  In both plots, the red curve is the numerically computed holographic screen entanglement entropy of $A_\psi$.  The dashed green curve is the static sphere approximation for $S(A\psi)$ which becomes more accurate at later $\tau$ (smaller $z$).  The orange curve with a sharp peak is $S_\textrm{Page}(A_\psi)$ as defined by equation \ref{thermal_entanglement} and the black curve is $S_\textrm{extensive}(A_\psi)$.  The horizontal line, provided for scale, marks the value of $\pi \alpha^2/2$ which is precisely one fourth of the extensive entropy of the de Sitter horizon.}
\label{fig_CC_dust}
\end{figure}

As explained above, the cosmological spacetimes we have been discussing have vacuum energy $\rho_\Lambda$ as well as some matter content that dilutes at late time.  The simplest case of this is when the universe is flat ($f(\chi) = \chi$) and when the additional matter content consists of only one species with density $\rho_\textrm{matter}$ and pressure $p_\textrm{m} = w \rho_\textrm{m}$.  The scale factor for this case is
\begin{equation}
a(\tau) = C \sinh \left[ {3  (1+w)  \tau \over 2  \alpha }  \right]^{2\over 3(1+w)}
\end{equation}
where the normalization factor $C$ is independent of $\tau$.

This setting is very useful to test the theoretical apparatus developed in this section.  In the case of $w=0$, figure \ref{fig_CC_dust} shows a variety of quantities we have discussed.  Figure \ref{fig_CC_dust} (a) and (b) depict the case of an earlier and later time leaf with $z/\alpha \approx .05$ and $z/\alpha \approx .02$ respectively.  The solid red curves show $S(A_\psi)$ (computed numerically) while the green curves give the static sphere approximation of equation \ref{dS_entropy_static}.  The dotted horizontal line marks half of the de Sitter entropy: $S_{1/2} =  \pi \alpha^2/2$.  As expected, $S(A_\psi) < S_{1/2}$. The orange curve with a discontinuity in its derivative is $S_\textrm{Page}^\infty(A_\psi)$.  Comparing figures \ref{fig_CC_dust} (a) and (b), one can see that $S(A_\psi)$ is approaching $S_\textrm{Page}^\infty(A_\psi)$ as $\tau \to \infty$.  Finally, the black curves shows extensive entropy: $S_\textrm{extensive}(A_\psi) = (1/4) \textrm{area}(A_\psi)$.  Note that $S(A_\psi) < S_\textrm{extensive}(A_\psi)$ for all $\psi$ as required by corollary \ref{thermal_bound_thm}.

\subsection*{Closed Universe with a Big Crunch}
\begin{figure}
\centering
\hspace*{-.72cm}   
\includegraphics[width=10cm]{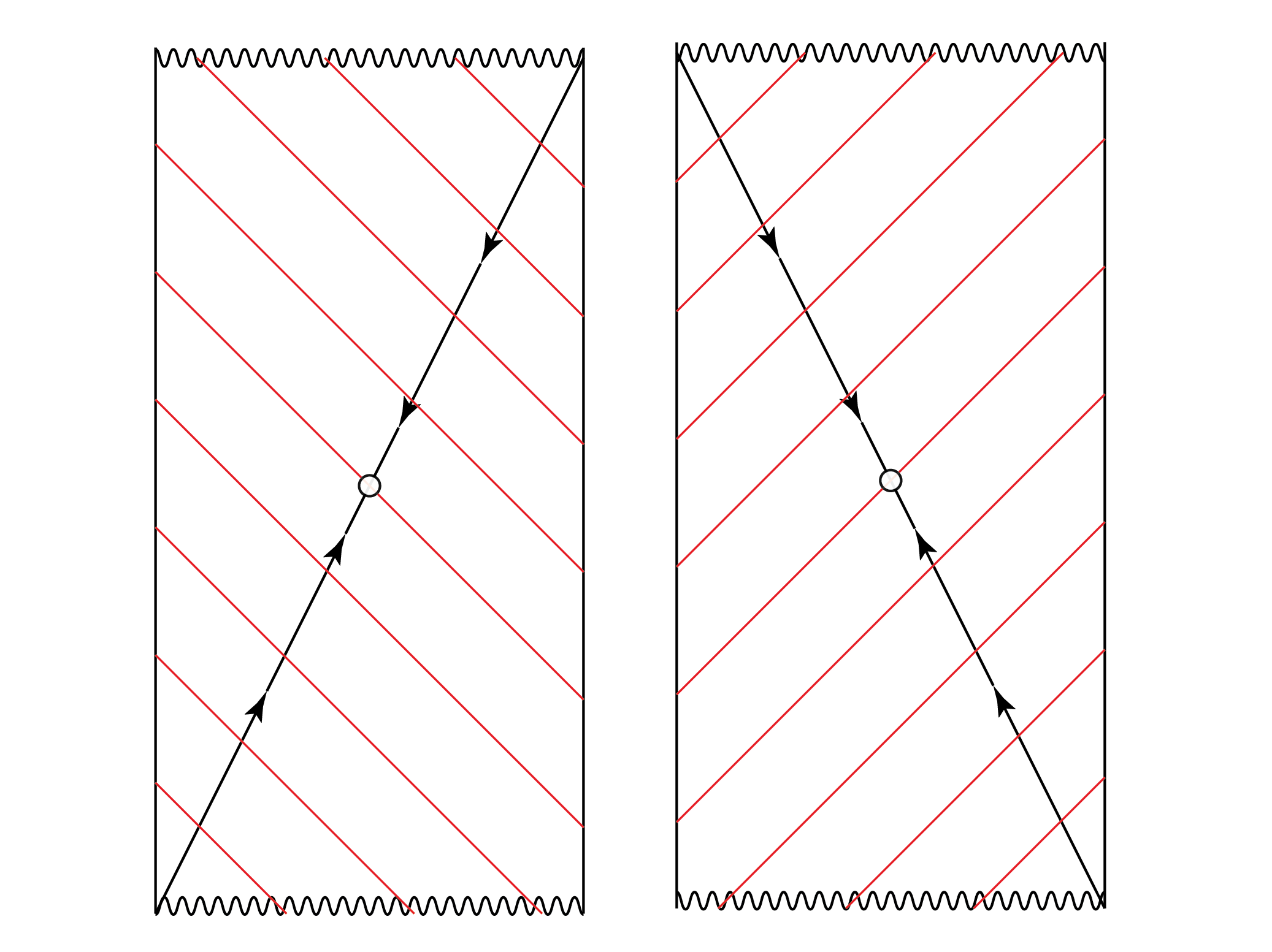}
\caption{Both Penrose diagrams here are for the same spacetime: a closed FRW universe with dust.  The red lines denote a null foliation and the black diagonals are the past and future holographic screens corresponding to the foliation.  The two figures demonstrate that different foliations give rise to different screens.  In both figures, the lower half of the diagonal is a past holographic screen and the upper half is a future holographic screen.  Arrows show the direction of increasing area.}
\label{fig_big_crunch}
\end{figure}

The holographic screen entanglement entropy structure of a closed universe with a past and future singularity is similar to that of approximate de Sitter space.  The spacetimes we consider have the metric of equation \ref{FRW_metric} with $f(\chi) = \sin(\chi)$.  In this case the coordinate $\chi$ takes values from $0$ to $\pi$.  We put one species of matter content in the spacetime that satisfies $p = w \rho$ which gives rise to a big bang at $\tau=0$ as well as a big crunch.  As before, we introduce a conformal time coordinate $\eta$ in terms of which the scale factor is
\[
a(\eta) = c  \left(\sin\frac{\eta}{q}\right)^{q}
\]
where $q = 2/(1+3w)$ and $c$ is constant.  This shows that the Penrose diagram for this spacetime is a rectangle with a time-to-space aspect ratio of $q$.  

Figure \ref{fig_big_crunch} shows the holographic screen structure of this spacetime for two examples of null foliations.  We focus on the diagram to the left in which case the null foliation (partially) consists of past light cones of a comoving worldline at the $\chi=0$.  As suggested by the figure, the holographic screen is given by
\[
\chi_\textrm{screen} = \frac{1}{q} \eta.
\]
However, a subtlety arises because the screen is a past holographic screen for $\eta <q \pi/2$ and a future screen for $\eta > q \pi/2$.  The sphere that connects the past and future screen is extremal (this was called an ``optimal'' surface in \cite{Bousso:1999cb}) and has area $4 \pi c^2$.  Let $\sigma(\eta)$ denote the leaf at conformal time $\eta$.  We put $\sigma_0 = \sigma(\eta = q \pi/2)$.  

Just as in the de Sitter case, this example leads to a saturation of the Page bound of equation \ref{thermal_bound} as leaves are maximized in area.  More precisely, if $A\subset \sigma(\eta)$, then $ \lim_{\eta \to q \pi /2} S(A) = S_\textrm{Page}^\infty(A)$ where in this case
\[
S_\textrm{Page}^\infty(A) =
\begin{cases}
\frac{1}{4} \textrm{area}(A) &  \textrm{area}(A) \leq  \frac{1}{2} \pi c^2 \\
\frac{1}{4}\left(4 \pi c^2 - \textrm{area}(A)\right) & \textrm{area}(A) > \frac{1}{2} \pi c^2.
\end{cases}
\]

It appears that $S(A)$ saturates the Page bound in a great variety of cases where the areas of leaves are bounded above.

\section{Concluding Remarks}
\label{sec_conclusion}
The proposal we have given above may open the door to a new research program: the study of the entanglement structure of general spacetimes.  In light of this, and for the sake of clarity, we now summarize the recipe for computing von Neumann entropy under the assumption of the screen entanglement conjecture discussed in section \ref{section2}:

\begin{enumerate}
\item Select a particular null foliation $\{ N_r \}$ of a spacetime with dimension $d$.
\item Find the codimension 2 surfaces $\{ \sigma_r \}$ with $\sigma_r \subset N_r$ that have maximal area on each $N_r$.
\item Take a $d-2$ dimensional subregion $A \subset \sigma_r$ with a boundary $\partial A$.
\item Of all extremal surfaces anchored to $\partial A$ and lying in the causal region $D_\sigma$ (see section \ref{section2}), select the one of minimal area.  The conjectured entropy $S(A)$ is then one fourth the area of the minimal extremal surface in Planck units.
\end{enumerate}

Potential applications of our conjecture are numerous.  One example not considered above is case of a spacetime with a black hole.  Black holes formed through the collapse of matter possess future holographic screens in their interiors that approach their horizons at late times.  It is of potential significance to investigate the entanglement structure of such spacetimes.  Perhaps such an analysis will shed light on the firewall paradox \cite{Almheiri:2012rt}.

If the screen entanglement conjecture is correct, it should still only be regarded as a leading order prescription for the computation of von Neumann entropies.  A version of the analysis of \cite{Faulkner:2013ana} may be extendible to the context of holographic screens.  It is not completely obvious how this should be done.  If $A$ is a region in a leaf $\sigma$ lying on a Cauchy slice $S_0$, one may consider the region on $S_0$ bounded by $A$ and its extremal surface $\ext (A)$ and compute the entanglement entropy of this region in a quantum field theory on the spacetime background.  On the other hand, it may be necessary to modify the spacetime position of the holographic screen itself as was done in \cite{Bousso:2015eda}.

\vskip .5cm
\indent {\bf Acknowledgments} 

We are very grateful to C. Akers, R. Bousso, Z. Fisher, D. Harlow, J. Koeller, M. Moosa, M. Van Raamsdonk, and A. Wall for for impactful discussions.  We thank R. Bousso and A. Wall for their comments on earlier versions of this work.  The work of FS is supported in part by the DOE NNSA Stewardship Science Graduate Fellowship.  The work of SJW is supported in part by the BCTP Brantley-Tuttle Fellowship for which he would like to extend his gratitude to Lynn Brantley and Douglas Tuttle.

\end{document}